\documentclass{article}
\usepackage{graphicx,epsfig,color}
\usepackage[shortend,ruled]{algorithm2e}
\usepackage{ifthen,verbatim,array,multirow}
\usepackage{latexsym,amssymb,amsmath,mathrsfs}
\usepackage{algorithm2e}
\usepackage{xspace}
\usepackage{setspace}
\usepackage{stackengine}
\usepackage{pgf}

\usepackage{textcomp}
\usepackage{stmaryrd}
\usepackage{subfigure}
\usepackage{hyperref}
\usepackage{setspace}
\usepackage{enumitem}
\newtheorem{proposition}{Proposition}

\newtheorem{definition}{Definition}
\newtheorem{example}{Example}
\newtheorem{remark}{Remark}
\newtheorem{observation}{Observation}

\def\algorithmautorefname{Algorithm}

{\noindent\emph{Proof sketch.~}}{\mbox{}\nobreak\hfill\hspace{6pt}\qed}

\newcommand\argvdash[1]{
  \mathrel{\stackengine{0.25ex}{\vdash}{\:\scriptstyle#1}{U}{c}{F}{T}{L}}
}

\let\implies=\undefined%
\newcommand{\implies}{\rightarrow}

\newcommand{\Implies}{\Rightarrow}
\newcommand{\satisfies}{\argvdash{s}}
\newcommand{\falsifies}{\argvdash{u}}

\newcommand{\tn}{\textnormal}

\newcommand{\todoF}[2]{}

\newcommand{\fml}[1]{{\mathcal{#1}}}

\newcommand{\mxsat}{MaxSAT\xspace}
\newcommand{\mnsat}{MinSAT\xspace}
\newcommand{\mxf}{MaxFalse\xspace}

\newcommand{\Pred}{\mathit{P}}

\definecolor{gray}{rgb}{.4,.4,.4}
\definecolor{midgrey}{rgb}{0.5,0.5,0.5}
\definecolor{darkgrey}{rgb}{0.3,0.3,0.3}
\definecolor{darkred}{rgb}{0.7,0.1,0.1}
\definecolor{lmidblue}{rgb}{0.3,0.3,0.9}
\definecolor{midblue}{rgb}{0.2,0.2,0.7}
\definecolor{darkblue}{rgb}{0.1,0.1,0.5}
\definecolor{defseagreen}{cmyk}{0.69,0,0.50,0}

\newcommand{\jnoteF}[1]{}

\newcounter{Comment}[Comment]
\DeclareMathOperator*{\nentails}{\nvDash} 
\DeclareMathOperator*{\entails}{\vDash} 
\DeclareMathOperator*{\lequiv}{\leftrightarrow}

\newboolean{extended}      
\setboolean{extended}{false} %
\newcommand{\iflongpaper}[1]{\ifthenelse{\boolean{extended}}{#1}{}}
\newcommand{\ifregularpaper}[1]{\ifthenelse{\boolean{extended}}{}{#1}}

\newboolean{addthanks}      
\setboolean{addthanks}{false} %
\newcommand{\ifaddack}[1]{\ifthenelse{\boolean{addthanks}}{#1}{}}

\newenvironment{proof}{\noindent\textbf{Proof.~}}{\mbox{}\nobreak\hfill\hspace{6pt}\\[-5pt]}
\newcommand{\proofitem}[1]{\newline\noindent\textit{#1.\xspace}}

\newcommand{\var}{\mathsf{var}}
\newcommand{\SAT}{\mathsf{SAT}} 

\newcommand{\CNF}{\mathbb{C}}
\newcommand{\DNF}{\mathbb{D}}
\newcommand{\BFML}{\mathbb{F}}
\newcommand{\ASGN}{\mathbb{A}}
\newcommand{\LITS}{\mathbb{L}}
\newcommand{\outc}{\mathsf{st}}
\newcommand{\maps}{\le_p}
\newcommand{\sres}{\textsf{sat}\xspace}
\newcommand{\ures}{\textsf{unsat}\xspace}
\newcommand{\lform}{$\mathscr{L}$\xspace}
\newcommand{\pform}{$\mathscr{P}$\xspace}

\newcommand{\bform}{$\mathscr{B}$\xspace}

\setlist{nolistsep}

\hypersetup{
  linktoc=all,
  colorlinks,
  citecolor=lmidblue, filecolor=lmidblue, linkcolor=lmidblue, urlcolor=lmidblue
}

\stackMath

\title{\bf\Large Computing Minimal Sets on Propositional Formulae I: \\
  Problems \& Reductions
}

\author{Joao Marques-Silva$^1$ and Mikol\'a\v{s} Janota$^2$
\\
$^1$ {CASL, University College Dublin, Belfield, Dublin, Ireland}\\
$^2$ {IST/INESC-ID, Lisbon, Portugal}
}
\date{}
\begin{document}
\maketitle
%
%
%

\begin{abstract}
\small    
%
Boolean Satisfiability (SAT) is arguably the archetypical 
NP-complete decision problem. Progress in SAT solving algorithms has
motivated an ever increasing number of practical applications in
recent years. 
However, many practical uses of SAT involve solving function 
as opposed to decision problems.
Concrete examples include computing minimal unsatisfiable subsets,
minimal correction subsets, prime implicates and implicants, minimal
models, backbone literals, and autarkies, among several others.
In most cases, solving a function problem requires a number of
adaptive or non-adaptive calls to a SAT solver. Given the
computational complexity of SAT, it is therefore important to develop 
algorithms that either require the smallest possible number of calls
to the SAT solver, or that involve simpler instances.
This paper addresses a number of representative function problems
defined on Boolean formulas, and shows that all these function
problems can be reduced to a generic problem of computing a minimal
set subject to a monotone predicate. This problem is referred to as
the Minimal Set over Monotone Predicate (MSMP) problem.
This exercise provides new ways for solving well-known function
problems, including prime implicates, minimal correction subsets,
backbone literals, independent variables and autarkies, among several
others.
Moreover, this exercise motivates the development of more efficient
algorithms for the MSMP problem.
Finally the paper outlines a number of areas of future research
related with extensions of the MSMP problem.
\end{abstract}


\clearpage
\tableofcontents
\clearpage
%
%
%

\section{Introduction} \label{sec:intro}

The practical success of Propositional Satisfiability (SAT) solvers
has motivated an ever increasing range of applications. Many
applications are naturally formulated as decision procedures. In
contrast, a growing number of of applications involve solving function
and opposed to decision problems.
Representative examples include computing maximum satisfiability,
minimum satisfiability, minimal unsatisfiable subsets of clauses,
minimal correction subsets of clauses, minimal and maximal models, the
backbone of a formula, the maximum autark assignment, prime implicants
and implicates, among many others.

In recent years, the most widely used approach for solving a
comprehensive range of function problems defined on Boolean formulas
consists of using the SAT solver as an oracle, which is called a
number of times polynomial in the size of the problem
representation. It is interesting to observe that this approach
matches in some sense well-known query complexity characterizations of
functions problems~\cite{garey-bk79,papadimitriou-bk94}. 

Moreover, as shown in~\cite{bradley-fmcad07,bradley-fac08,msjb-cav13},
representative function problems can be viewed as computing a minimal
set given some monotone predicate, and develop a number of algorithms
for computing a minimal set subject to a monotone predicate.
This work motivates the question of whether more function problems can
be represented as computing a minimal set subject to a monotone
predicate.
The implications could be significant, both in terms of new algorithms
for different problems, as well as possible new insights into some of
these problems.
This paper addresses this question, and shows that a large number of
function problems defined on Boolean formulae can be cast as computing
a minimal set subject to a monotone predicate.
As a result, all the algorithms developed in~\cite{msjb-cav13} for
some function problems can also be used for this much larger set of
function problems.

The paper is organized as follows. \autoref{sec:prelim} introduces the
notation and definitions used throughout the paper.
\autoref{sec:probs} defines and overviews the function problems
studied in later sections.
\autoref{sec:maps} shows how all the functions described in
\autoref{sec:probs} can be reduced to solving an instance of the more
general MSMP problem.
\autoref{sec:conc} concludes the paper, by summarizing the main
contributions and outlining a number of research directions.


%
%
%

\section{Preliminaries} \label{sec:prelim}

This section introduces the notation used in the remainder of the
paper, covering propositional formulae, monotone predicates,
complexity classes, and problem reductions.

\subsection{Propositional Formulae} \label{ssec:bf}

Standard propositional logic definitions are used throughout the paper
(e.g.\ \cite{buning-bk99,sat-handbook09}), some of which are reviewed
in this section.

Sets are represented in caligraphic font, e.g.\ $\fml{R}, \fml{T},
\fml{I}, \ldots$
Propositional formulas are also represented in caligraphic font,
e.g.\ $\fml{F}, \fml{H}, \fml{S}, \fml{M}, \ldots$ 
Propositional variables are represented with letters from the end of
the alphabet, e.g.\ $x, w, y, z$, and indeces can be used,
e.g.\ $x_1, w_1, \ldots$
An atom is a propositional variable.
A literal is a variable $x_i$ or its complement $\neg x_i$.
A propositional formula (or simply a formula) $\fml{F}$ is defined
inductively over a set of propositional variables, with the standard
logical connectives, $\neg$, $\land$, $\lor$, as follows:
\begin{enumerate}
\item An atom is a formula.
\item If $\fml{F}$ is a formula, then $(\neg\fml{F})$ is a formula.
\item If $\fml{F}$ and $\fml{G}$ are formulas, then
  $(\fml{F}\lor\fml{G})$ is a formula.
\item If $\fml{F}$ and $\fml{G}$ are formulas, then
  $(\fml{F}\land\fml{G})$ is a formula. 
\end{enumerate}
The inductive step could be extended to include other propositional
connectives, e.g.\ $\implies$ and $\lequiv$. (The use of parenthesis
is not enforced, and standard binding rules apply
(e.g.\ \cite{buning-bk99}), with parenthesis being used only to
clarify the presentation of formulas.) 
The inductive definition of a propositional formula allows associating
a parse tree with each formula, which can be used for evaluating truth
assignments.
The variables of a propositional formula $\fml{F}$ are represented by
$\var(\fml{F})$. For simplicitly, the set of variables of a formula
will be denoted by $X\triangleq\var(\fml{F})$.
A clause $c$ is a non-tautologous disjunction of literals. A term $t$
is a non-contradictory conjunction of literals.
Commonly used representations of propositional formulas include conjunctive
and disjunctive normal forms (resp.\ CNF and DNF).
A CNF formula $\fml{F}$ is a conjunction of clauses.
A DNF formula $\fml{F}$ is a disjunction of terms.
CNF and DNF formulas can also be viewed as sets of sets of literals. 
Both representations will be used interchangeably throughout the
paper.
In the remainder of the paper, propositional formulas are referred to as
formulas, and this can either represent an arbitrary propositional formula,
a CNF formula, or a DNF formula. The necessary qualification will be
used when necessary.
The following sets are used throughout. $\BFML$ denotes the set of
propositional formulas, $\CNF\subset\BFML$ denotes the set of CNF
formulas, and $\DNF\subset\BFML$ denotes the set of DNF formulas. 

When the set of variables of a formula $\fml{F}$ is relevant, the
notation $\fml{F}[X]$ is used, meaning that $F$ is defined in terms of 
variables from $X$. 
Replacements of variables will be used. The notation
$\fml{F}[x_i/y_i]$ represents formula $\fml{F}$ with variable $x_i$
replaced with $y_i$. This definition can be extended to more than one
variable. For the general case, if $X=\{x_1,\ldots,x_n\}$ and
$Y=\{y_1,\ldots,y_n\}$, then the notation $\fml{F}[X/Y]$ represents
formula $\fml{F}$ with $x_1$ replaced by $y_1$, $x_2$ replaced by
$y_2$, ..., and $x_n$ replaced by $y_n$. Alternatively, one could
write $\fml{F}[x_1/y_1,x_2/x_2,\dots,x_n/y_n]$.
Given $v_i\in\{0,1\}$, $\fml{F}[x_i/v_i]$ represents formula $\fml{F}$
with variable $x_i$ replaced with $v_i$. Alternatively, the notation
$\fml{F}_{x_i=v_i}$ can be used.

The paper addresses mainly \emph{plain} CNF formulas, i.e.\ formulas
    where any clause can be dropped (or relaxed) from the formula,
    being referred to as {\em soft} clauses~\cite{manya-hdbk09}.
Nevertheless, in some settings, CNF formulas can have {\em hard}
clauses, i.e.\ clauses that cannot be dropped (or relaxed).
In these cases, $\fml{F}$ can be viewed as a 2-tuple
$\langle\fml{H},\fml{B}\rangle$, where $\fml{H}$ denotes the hard 
clauses, and $\fml{B}$ denotes the {\em soft} (or relaxable, or 
breakable) clauses. 
(Observe that any satisfiability test involving $\fml{F}$ requires the
hard clauses to be satisfied, whereas some of the clauses in $\fml{B}$
may discarded.) 
Moreover, weights can be associated with the soft clauses as
follows. A weight function
$\omega:\fml{H}\cup\fml{B}\to\mathbb{R}\cup\{\top\}$ associates a  
weight with each clause. For any soft clause $c\in\fml{B}$,
$\omega(c)\not=\top$, whereas for any hard clause $c\in\fml{H}$,
$\omega(c)=\top$, where $\top$ is such that it exceeds
$\sum_{c\in\fml{B}}\,\omega(c)$, meaning that hard clauses are too
costly to falsify.
Throughout the paper, and unless otherwise stated, it is assumed that
either $\fml{H} = \emptyset$ or $\fml{B} = \emptyset$ , i.e.\ all
clauses are soft or all clauses are hard, and so $\fml{F}$ corresponds
to $\fml{B}$ or to $\fml{H}$, respectively. Moreover, the (implicit)
weight function assigns cost 1 to each (soft) clause.
Finally, the above definitions can be extended to handle groups of
clauses, e.g.~\cite{liffiton-jar08,hmms-aicomm13}.

Given a formula $\fml{F}$, a truth assignment $\nu$ is a
map from the variables of $\fml{F}$ to $\{0, 1\}$,
$\nu:\var(\fml{F})\to\{0,1\}$.
Given a truth assignment $\nu$, the value taken by a formula, denoted
$\fml{F}^{\nu}$ is defined inductively as follows:
\begin{enumerate}
\item If $x$ is a variable, $x^{\nu} = \nu(x)$.
\item If $\fml{F} = (\neg\fml{G})$, then
  \[\fml{F}^{\nu} = \left\{
  \begin{array}{lcl}
    0 & & \tn{if $\fml{G}^{\nu} = 1$} \\
    1 & & \tn{if $\fml{G}^{\nu} = 0$} \\
  \end{array}\right.\]
\item If $\fml{F} = (\fml{E}\lor\fml{G})$, then
  \[\fml{F}^{\nu} = \left\{
  \begin{array}{lcl}
    1 & & \tn{if $\fml{E}^{\nu} = 1$ or $\fml{G}^{\nu} = 1$} \\
    0 & & \tn{otherwise} \\
  \end{array}\right.\]
\item If $\fml{F} = (\fml{E}\land\fml{G})$, then 
  \[\fml{F}^{\nu} = \left\{
  \begin{array}{lcl}
    1 & & \tn{if $\fml{E}^{\nu} = 1$ and $\fml{G}^{\nu} = 1$} \\
    0 & & \tn{otherwise} \\
  \end{array}\right.\]
\end{enumerate}
The inductive step could be extended to include additional
propositional connectives, e.g.\ $\implies$ and $\lequiv$.
A truth assignment $\nu$ such that $\fml{F}^{\nu} = 1$ is referred to
as a \emph{satisfying truth assignment}.
A formula $\fml{F}$ is \emph{satisfiable} if it has a satisfying truth
assignment; otherwise it is \emph{unsatisfiable}.
As a result, the problem of propositional satisfiability is defined as
follows: 
\begin{definition}[Propositional Satisfiability (SAT)]
  Given a formula $\fml{F}$, the decision problem SAT consists of
  deciding whether $\fml{F}$ is satisfiable.
\end{definition}

The standard semantic entailment notation is used throughout. Let 
$\fml{A},\fml{C}\in\BFML$. $\fml{A}\entails\fml{C}$ denotes that for
every truth assignment $\nu$ to the variables in
$\var(\fml{A})\cup\var(\fml{C})$, $(\fml{A}^{\nu}=1)\Implies(\fml{C}^{\nu}=1)$. 
The equivalence notation $\fml{A}\equiv\fml{C}$ is used to denote that 
$\fml{A}\entails\fml{C}\land\fml{C}\entails\fml{A}$, indicating that
$\fml{A}$ and $\fml{C}$ have the same satisfying truth assignments,
when $\var(\fml{A})=\var(\fml{C})$.
If a formula $\fml{F}\in\BFML$ is satisfiable, we write
$\fml{F}\nentails\bot$. If a formula $\fml{F}$ is unsatisfiable, we
write $\fml{F}\entails\bot$. 
Moreover, if $\nu$ is a satisfying truth assignment of $\fml{F}$, the
notation $\nu\entails\fml{F}$ is also used. 
Finally, if $\fml{F}$ is a tautology, then the notation
$\top\entails\fml{F}$ is used, whereas $\top\nentails\fml{F}$ denotes
that $\fml{F}$ is not a tautology.

The following results are well-known, e.g.\ \cite{buning-bk99}, and
will be used throughout.

\begin{proposition} \label{prop:satsubset}
  Let $\fml{F}\in\CNF$, with $\fml{F}\nentails\bot$. Then,
  $\forall_{\fml{E}\in\CNF},(\fml{E}\subseteq\fml{F})\Implies(\fml{E}\nentails\bot)$.
\end{proposition}

\begin{proposition} \label{prop:usatsupset}
  Let $\fml{U}\in\CNF$, with $\fml{U}\entails\bot$. Then,
  $\forall_{\fml{T}\in\CNF},(\fml{T}\supseteq\fml{U})\Implies(\fml{T}\entails\bot)$.
\end{proposition}

\begin{proposition} \label{prop:ntautsubset}
  Let $\fml{F}\in\DNF$, with $\top\nentails\fml{F}$. Then,
  $\forall_{\fml{E}\in\DNF},(\fml{E}\subseteq\fml{F})\Implies(\top\nentails\fml{E})$.
\end{proposition}

\begin{proposition} \label{prop:tautsupset}
  Let $\fml{U}\in\DNF$, with $\top\entails\fml{U}$. Then,
  $\forall_{\fml{T}\in\DNF},(\fml{T}\supseteq\fml{U})\Implies(\top\entails\fml{T})$.
\end{proposition}

Given a formula $\fml{F}$, with set of variables $X = \var(\fml{F})$,
the following additional definitions apply. 
$\LITS(X)\triangleq\{x,\neg x\,|\,x\in X\}$ represents the set of
literals given the variables in $X$. 
A truth assignment $\nu$ can be represented as a set of literals
$\fml{V}\subseteq\LITS$, interpreted as a conjunction of literals or a
term, where each literal in $\fml{V}$ encodes the value assigned to a
given variable $x\in X$, i.e.\ literal $x$ if $\nu(x)=1$ and literal
$\neg x$ if $\nu(x)=0$. 
Clearly, $|\fml{V}| = |X|$. 
In the remainder of the paper, an assignment can either be represented
as a map or as a set of literals as defined above. This will be clear
from the context.

Sets of literals are used to represent partial truth assignments. The
set of all partial truth assignments $\ASGN$, given
$X\triangleq\var(\fml{F})$, is defined by $\ASGN(X)\triangleq\{\fml{V}\subseteq\LITS(X)\:|\:(\forall_{x\in X)},(x\not\in\fml{V})\lor(\neg x\not\in\fml{V}))\}$.
%
(For simplicitly, and when clear from the context, the dependency of
$\LITS$ and $\ASGN$ with $X$ will be omitted.)
Sets of literals can also be used to satisfy or falsify CNF or DNF
formulas. For a set of literals $\fml{V}\in\ASGN$ and
$\fml{F}\in\BFML$, 
the notation $\fml{V}\satisfies\fml{F}$ denotes whether assigning 1 to
the literals in $\fml{V}$ satisfies $\fml{F}$. Similarly, the notation
$\fml{V}\falsifies\fml{F}$ denotes whether assigning 1 to the literals
in $\fml{V}$ falsifies $\fml{F}$.

As noted above, sets of literals are in most cases interpreted as the
conjunction of the literals, e.g.\ as a term. However, in some
situations, it is convenient to interpret a set of literals as a
disjunction of the literals, e.g.\ as a clause.
Throughout the paper, the following convention is used.
A set of literals qualified as a term, as a truth assignment, or as an
implicant (defined below) is interpreted as a conjunction of literals.
A set of literals qualified as a clause, or as an implicate (defined
below) is interpreted as a disjunction of literals.
Given $\fml{F}\in\BFML$, a term $t\in\ASGN$ is an \emph{implicant} of 
$\fml{F}$ iff $t\entails\fml{F}$. 
(Alternatively, we could write $t\in\ASGN$ is an \emph{implicant} of
$\fml{F}$ iff  $\land_{l\in t}\,(l)\entails\fml{F}$.)
Similarly, a clause $c\in\ASGN$ is an \emph{implicate} of $\fml{F}$
iff $\fml{F}\entails c$. 
(Alternatively, we could write $c\in\ASGN$ is an \emph{implicate} of
$\fml{F}$ iff $\fml{F}\entails(\lor_{l\in c}\,l)$.)

In some settings it is necessary to reason with the variables assigned
value 1. This is the case for example when reasoning with {\em
  minimal} and {\em maximal models}.
Given a truth assignment, $\nu$, with $\fml{V}$ the associated set of
literals, the variables assigned value 1 are given by
$\fml{M}=\fml{V}\cap X$. 
The function $\nu(M,X)$ allows recovering the truth assignment
associated with a set $\fml{M}$ of variables assigned value 1. 
If the assignment $\nu$, associated with a set $\fml{M}$ of variables
assigned value 1, is satisfying, then $\fml{M}$ is referred to as a
{\em model}.
Moreover, the notation $\nu(M,X)\entails\fml{F}$ is used to denote
that, given a set $\fml{M}$ of variables assigned value 1, the
associated truth assigment is satisfying. In contrast,
$\nu(M,X)\nentails\fml{F}$ is used to denote that the truth assignment
falsifies the formula.


Most of the algorithms described in this paper use sequences of calls
to a SAT solver. A SAT solver accepts a (CNF-encoded) propositional formula
$\fml{F}$ as a single argument and returns a 2-tuple $(\outc,
\alpha)$, i.e.\ $(\outc, \nu)=\SAT(\fml{F})$, where $\outc\in\{0,1\}$
denotes whether the formula is unsatisfiable ($\outc=0$, or simply
\ures) or satisfiable ($\outc=1$, or simply \sres), and $\nu$ is a
(satisfying) truth assignment in case the formula is satisfiable. 
$\nu$ will also be referred to as a \emph{witness} of satisfiability.
For simplicity, in this paper it is assumed the SAT solver does not
return unsatisfiable subformulas when the outcome is \ures, although
this feature is available in many modern SAT
solvers, e.g.\ \cite{een-sat03,biere-jsat08}.
Throughout the paper, it will be implicitly assumed that a SAT solver
call yields not only the 0/1 outcome, but the actual 2-tuple
$(\outc,\nu)$.

Modern SAT solvers typically accept CNF formulas~\cite{mslm-hdbk09}. 
Procedures for CNF-encoding (or clausifying) arbitrary propositional
formulas are well-known (e.g.\ \cite{tseitin68,plaisted-jsc86}).
Throughout the paper the propositional formulas passed to SAT solvers are
often not in CNF. For simplicity, it is left implicit that a
clausification procedure would be invoked if necessary.
Moreover, additional non-clausal constraints will be used. These
include pseudo-Boolean constraints and cardinality constraints,
e.g.~\cite{manquinho-hdbk09,prestwich-hdbk09,een-jsat06}.
Examples of clausification approaches are described
in~\cite{prestwich-hdbk09,een-jsat06}.

\subsection{Monotone Predicates \& MSMP} \label{ssec:mp}

A predicate $\Pred:2^{\fml{R}}\rightarrow\{0,1\}$, defined on
$\fml{R}$, is said to be {\em monotone}~(e.g.\ \cite{bradley-fmcad07})
if whenever $\Pred(\fml{R}_0)$  holds, with
$\fml{R}_0\subseteq\fml{R}$, then $\Pred(\fml{R}_1)$ also holds, with
$\fml{R}_0\subseteq\fml{R}_1\subseteq\fml{R}$.
Observe that $\Pred(\fml{R})$ can be assumed, but this is not
required. Also, $\Pred(\fml{R})$ can be tested with a single predicate
test. Moreover, observe that, if there exists a set
$\fml{R}_0\subseteq\fml{R}$ such that $\Pred(\fml{R}_0)$ holds, and
$\Pred$ is monotone, then $\Pred(\fml{R})$ also holds.

\begin{definition} \label{def:mp}
  Let $\Pred$ be a monotone predicate, and let
  $\fml{M}\subseteq\fml{R}$ such that $\Pred(\fml{M})$ holds.
  $\fml{M}$ is minimal iff $\forall_{\fml{M}'\subsetneq\fml{M}},\neg
  \Pred(\fml{M}')$.
\end{definition}

\begin{example}
  It is simple to conclude that, given a finite set $\fml{W}$,
  predicate $\Pred(\fml{W})\triangleq[\,|\fml{W}|\ge K\,]$, for some
  $K\ge 0$, is monotone.
  In contrast, predicate $\Pred(\fml{W})\triangleq[\,|\fml{W}|\mod 2 =
    1\,]$ is not monotone. 
\end{example}

\begin{definition}[MSMP Problem]
Given a monotone predicate $\Pred$, the {\em Minimal Set over a
  Monotone Predicate} (MSMP) problem consists in finding a minimal 
subset $\fml{M}$ of $\fml{R}$ such that $\Pred(\fml{M})$ holds.
\end{definition}

Monotone predicates were used in~\cite{bradley-fmcad07,bradley-fac08}
for describing algorithms for computing a prime implicate of CNF
formula given a clause.
The MSMP problem was introduced in~\cite{msjb-cav13}.
\cite{msjb-cav13} also showed that a number of additional function
problems defined on propositional formulas could be reduced to the
MSMP problem, including the function problems of computing minimal
unsatisfiable subsets, minimal correction subsets and minimal models.

\subsection{Function Problems} \label{ssec:fp}

For many computational problems, the main goal is not to solve a
decision problem, but to solve instead a {\em function
  problem}~\cite[Section 10.3]{papadimitriou-bk94}\footnote{Function
  problems are also referred to as {\em search
    problems}~\cite[Section 5.1]{garey-bk79}.}.
The goal of a function problem is to compute some solution, e.g.\ a
satisfying truth assignment in the case of SAT, the size of a prime
implicate, the actual prime implicate, the largest number of clauses
that any truth assignment can satisfy in a CNF formula, etc.
%
%
This paper addresses function problems defined on propositional
formulas, which are defined in \autoref{sec:probs}.
The main focus of the paper are function problems that can be related
with computing a minimal set over a monotone predicate, i.e.\ function
problems that can be reduced to the MSMP problem.
Optimization problems will also be studied, with the purpose of
illustrating the modeling flexibility of monotone predicates.

\subsection{Problem Reductions}

The notation $A \maps B$ is used to denote that any instance of a
function problem $A$ can be solved by solving instead a
polynomially-related instance of another function problem $B$.
In this paper, $B$ denotes the MSMP problem, and any reduction $A
\maps B$ indicates that the solution to the function problem $A$ for a
concrete instance $a$ can be obtained by computing a minimal set
(subject to a monotone predicate) for some resulting concrete instance
$b$ of the MSMP problem.
Moreover, all reductions described in the paper are provided with
enough detail to make it straightfoward to conclude that the problem
instances are polynomially related.

\subsection{Related Work} \label{ssec:relw}

The work in this paper is motivated by the use of monotone predicates
for computing a prime implicate of a CNF
formula given an implicate (e.g.\ a
clause)~\cite{bradley-fmcad07,bradley-fac08}.
This work was recently extended to show that monotone predicates can
be applied to other problems, namely MUSes, MCSes, minimal models and
(also) prime implicates given a clause~\cite{msjb-cav13}.
\cite{msjb-cav13} also proposed the {\sc Progression} algorithm for
the MSMP problem. 

A tightly related concept in computational complexity is {\em
  hereditarity}~\cite{toda-sct93,toda-ic95}, which was proposed in the
90s to develop lower and upper bounds on the query complexity of
computing maximal solutions. A concrete example is the computation of
minimal unsatisfiability and maximal
satisfiability~\cite{toda-sct93,toda-ic95}. In contrast, other
function problems, e.g.\ minimal and maximal models, involve different 
definitions from those standard in AI
(e.g.\ \cite{ben-eliyahu-amai96}).
Moreover, one concern of this paper is to develop a rigorous
characterization of a SAT solver as an oracle. As shown
in~\autoref{sec:prelim}, if the time spent on the SAT solver is
ignored, then a SAT solver can viewed as a witness
oracle~\cite{buss-aptcc93}.

The paper addresses function problems defined on Boolean formulas,
which intersect a wide range of areas of research. References to
related work, both on the actual function problems, and associated
areas of research are included throughout the paper.


%
%
%

\section{Target Function Problems} \label{sec:probs}

This section provides a brief overview of the function problems 
studied in the rest of the paper. All function problems studied in
this section are defined on Boolean formulas.
Neverthless, the work can be extended to function problems defined on
expressive domains, e.g.\ ILP, SMT, CSP, etc.

\subsection{Problem Definitions} \label{ssec:fdefs} 

The function problems described below can be organized as follows: (i)
minimal unsatisfiability (and maximal satisfiability); (ii)
irredundant subformulas; (iii) maximal falsifiability (and minimal
satisfiability); (iv) minimal and maximal models; (v) prime implicates
and implicants; (vi) backbone literals; (vii) formula entailment;
(viii) variable independence; and (ix) maximum autarkies.
%
Afterwards, \autoref{sec:maps} shows that all these function problems
can be reduced to the MSMP problem.
In what follows, and unless otherwise stated, $\fml{F}$ denotes an
arbitrary Boolean formula. In some specific cases, $\fml{F}$ is
restricted to be in CNF (or in DNF), and this will be noted.

\subsubsection{Minimal Unsatisfiability \& Maximal Satisfiability}

Minimal unsatisfiability and maximal satisfiability have been
extensively studied in a number of contexts
(e.g.\ \cite{kullmann-hdbk09a,manya-hdbk09,blms-aicomm12,mhlpms-cj13}
and references therein). In this paper, the following definitions are
used.

\begin{definition}[MUS; FMUS] \label{def:mus}
  Let $\fml{F}\in\CNF$, with $\fml{F}\entails\bot$.
  $\fml{M}\subseteq\fml{F}$ is a {\em  Minimal
    Unsatisfiable Subset}
  (MUS) iff $\fml{M}\entails\bot$ and
  $\forall_{\fml{M}'\subsetneq\fml{M}},\,\fml{M}'\nentails\bot$.
  FMUS is the function problem of computing an MUS of $\fml{F}$.
\end{definition}

\begin{definition}[MCS; FMCS] \label{def:mcs}
  Let $\fml{F}\in\CNF$, with $\fml{F}\entails\bot$.
  $\fml{C}\subseteq\fml{F}$ is a {\em Minimal Correction Subset} (MCS)
  iff $\fml{F}\setminus\fml{C}\nentails\bot$ and
  $\forall_{\fml{C}'\subsetneq\fml{C}},\,\fml{F}\setminus\fml{C}'\entails\bot$.
  FMCS is the function problem of computing an MCS of $\fml{F}$.
\end{definition}

\begin{definition}[MSS; FMSS] \label{def:mss}
  Let $\fml{F}\in\CNF$, with $\fml{F}\entails\bot$.
  $\fml{S}\subseteq\fml{F}$ is a {\em Maximal Satisfiable Subset}
  (MSS) iff $\fml{S}\nentails\bot$ and
  $\forall_{\fml{F}\supseteq\fml{S}'\supsetneq\fml{S}},\,\fml{S}'\entails\bot$.
  FMSS is the function problem of computing an MSS of $\fml{F}$. 
\end{definition}

The relationship between MCSes and MSSes is well-known,
e.g.\ \cite{lozinskii-jetai03,liffiton-jar08}:
\begin{remark} \label{rm:mss}
  $\fml{C}$ is an MCS of $\fml{F}$ iff $\fml{S} =
  \fml{F}\setminus\fml{C}$ is an MSS of $\fml{F}$.
\end{remark}

The \mxsat problem is defined assuming unweighted clauses.
\begin{definition}[LMSS/\mxsat] \label{def:lmss}
  Let $\fml{F}\in\CNF$, with $\fml{F}\entails\bot$. 
  An MSS $\fml{S}$ of $\fml{F}$ is a largest MSS (LMSS) iff for any
  MSS $\fml{S}'$ of $\fml{F}$, $|\fml{S}'|\le|\fml{S}|$.
  FLMSS (or \mxsat) is the function problem of computing an LMSS of
  $\fml{F}$.
\end{definition}
\begin{observation}
  Observe that \mxsat is defined as the function problem of computing
  one largest MSS. In some contexts, other definitions are used,
  namely as the problem of computing the largest number of
  simultaneously satisfied clauses~\cite{krentel-jcss88}. This
  distinction is quite significant for the unweighted case of \mxsat.
  The definition used in this paper aims to model the function problem
  actually solved by modern \mxsat solvers. Since practical \mxsat
  solvers always compute a witness of the reported solution, this
  avoids the issue with reproducing a witness in the case an NP oracle
  were considered.
\end{observation}
%

\begin{definition}[SMCS] \label{def:smcs}
  Let $\fml{F}\in\CNF$, with $\fml{F}\entails\bot$. 
  An MCS $\fml{C}$ of $\fml{F}$ is a smallest MCS (SMCS) iff for any 
  MCS $\fml{C}'$ of $\fml{F}$, $|\fml{C}'|\ge|\fml{C}|$.
  FSMCS is the function problem of computing an SMCS of $\fml{F}$.
\end{definition}

\begin{remark}
  Clearly, $\fml{S}\subseteq\fml{F}$ is an LMSS of $\fml{F}$ iff
  $\fml{C} = \fml{F}\setminus\fml{S}$ is an SMCS of $\fml{F}$.
\end{remark}

Moreover, it is straightforward to extend the above definitions to
cases where there are hard clauses (i.e.\ clauses that cannot be
dropped from the formula) and where soft clauses have
weights~\cite{manya-hdbk09}.
%
%
There exists a large body of theoretical and practical work on
computing MUSes, MCSes and on solving \mxsat, including a well-known 
minimal hitting set duality relationship between MUSes and
MCSes~\cite{reiter-aij87,lozinskii-jetai03,liffiton-jar08}.
Recent references describing new algorithms and surveys for these
problems
include~\cite{mazure-ictai08,manya-hdbk09,ms-mvlsc12,blms-aicomm12,ansotegui-aij13,mshjpb-ijcai13,mhlpms-cj13}. 
There is a growing number of practical applications of both minimal
unsatisfiability and maximal satisfiability
(e.g.\ see~\cite{mazure-ictai08,ms-mvlsc12,blms-aicomm12,mhlpms-cj13}
and references therein). Recent examples of applications
include~\cite{veneris-fmcad07,majumdar-cav11,majumdar-pldi11,gurfinkel-cav13,jgms-vamos14,gbms-tacas14}.
Extraction of MUSes for non-clausal formulas and for SMT formulas is
addressed
in~\cite{kullmann-hdbk09a,bms-sat11,shen-csis11,gurfinkel-cav13,jgms-vamos14,gbms-tacas14}.

Finally, we should note that similar definitions can be developed for
DNF formulas, by considering \emph{minimal validity} instead of
\emph{minimal unsatisfiability}. This would also allow developing
several tightly related function problems, as above.

\subsubsection{Irredundant Subformulas}

The definitions for minimal unsatisfiability (see previous section)
can be reformulated for the case where the goal is to remove
redundancy from a CNF formula, as follows. 

\begin{definition}[MES; FMES] \label{def:mes}
  Let $\fml{F}\in\CNF$.
  $\fml{E}\subseteq\fml{F}$ is a {\em Minimal Equivalent Subset} (MES)
  iff $\fml{E}\equiv\fml{F}$ and
  $\forall_{\fml{E}'\subsetneq\fml{E}},\,\fml{E}'\not\equiv\fml{F}$.
  FMES is the function problem of computing an MES of $\fml{F}$.
\end{definition}

\begin{definition}[MDS; FMDS] \label{def:mds}
  Let $\fml{F}\in\CNF$. 
  $\fml{D}\subseteq\fml{F}$ is a {\em Minimal Distinguishing Subset}
  (MDS) iff $\fml{F}\setminus\fml{D}\not\equiv\fml{F}$ and
  $\forall_{\fml{D}'\subsetneq\fml{D}},\,\fml{F}\setminus\fml{D}'\equiv\fml{F}$.
  FMDS is the function problem of computing an MDS of $\fml{F}$.
\end{definition}

\begin{definition}[MNS; FMNS] \label{def:mns}
  Let $\fml{F}\in\CNF$.
  $\fml{N}\subseteq\fml{F}$ is a {\em
    Maximal Non-equivalent Subset} (MNS) iff
  $\fml{N}\not\equiv\fml{F}$ and
  $\forall_{\fml{F}\supseteq\fml{N}'\supsetneq\fml{N}},\,\fml{N}'\equiv\fml{F}$.
  FMNS is the function problem of computing an MNS of $\fml{F}$.
\end{definition}

\begin{remark} \label{rm:mns}
  $\fml{D}$ is an MDS of $\fml{F}$ iff
  $\fml{N}=\fml{F}\setminus\fml{D}$ is an MNS of $\fml{F}$.
\end{remark}

\begin{definition}[LMNS] \label{def:lmns}
  Let $\fml{F}\in\CNF$. 
  An MNS $\fml{N}$ of $\fml{F}$ is a largest MNS (LMNS) iff for any
  MNS $\fml{N}'$ of $\fml{F}$, $|\fml{N}'|\le|\fml{N}|$.
  FLMNS is the function problem of computing an LMNS of $\fml{F}$.  
\end{definition}

\begin{definition}[SMDS] \label{def:smds}
  Let $\fml{F}\in\CNF$. 
  An MDS $\fml{D}$ of $\fml{F}$ is a smallest MDS (SMDS) iff for any
  MDS $\fml{D}'$ of $\fml{F}$, $|\fml{D}'|\ge|\fml{D}|$.
  FSMDS is the function problem of computing an SMDS of $\fml{F}$.  
\end{definition}

\begin{remark}
  Clearly, $\fml{N}\subseteq\fml{F}$ is an LMNS of $\fml{F}$ iff
  $\fml{D} = \fml{F}\setminus\fml{N}$ is an SMDS of $\fml{F}$.
\end{remark}

Moreover, it is straightforward to extend the above definitions to
cases where there are hard clauses (i.e.\ clauses that cannot be
dropped from the formula) and where soft clauses have weights.

Complexity charaterizations of computing irredundant subformulas were
studied 
in~\cite{liberatore-aij05,liberatore-aij08a,liberatore-aij08b,kullmann-fi11b}. Recent
practical algorithms
include~\cite{roussel-aaai00,sais-ecai08,bjlms-cp12}. 

Finally, we should note that definitions of irredundancy can be
developed for DNF formulas in terms of a subset-minimal set of terms
equivalent to the original formula. As before, this would allow 
defining additional function problems, as above.

\subsubsection{Minimal Satisfiability \& Maximal Falsifiability}

In some settings, the goal is to minimize the number of satisfied
clauses. This is generally referred to as the {\em Minimum
  Satisfiability} (MinSAT) problem~\cite{kohli-sjdm94}. By analogy
with the MaxSAT case, one can also consider extremal
sets~\cite{impms-lpar13}.

\begin{definition}[All-Falsifiable]
  Let $\fml{F}\in\CNF$.
  $\fml{U}\subseteq\fml{F}$ is {\em All-Falsifiable} if there exists a
  truth assignment $\nu$, to $\var(\fml{U})$, such that $\nu$
  falsifies all clauses in $\fml{U}$.
\end{definition}

\begin{definition}[MFS; FMFS] \label{def:mfs}
  Given $\fml{F}\in\CNF$, a {\em Maximal Falsifiable Subset} (MFS)
  of $\fml{M}\subseteq\fml{F}$ is all-falsifiable and,
  $\forall_{\fml{F}\supseteq\fml{N}\supsetneq\fml{M}},\,$ $\fml{N}$ is
    not all-falsifiable.
  FMFS is the function problem of computing an MFS of $\fml{F}$.
\end{definition}

\begin{definition}[MCFS; FMCFS] \label{def:mcfs}
  Given $\fml{F}\in\CNF$, a {\em Minimal Correction (for
    Falsifiability) Subset} (MCFS) is a set $\fml{C}\subseteq\fml{F}$ 
  such that $\fml{F}\setminus\fml{C}$ is all-falsifiable and
  $\forall_{\fml{C}'\subsetneq\fml{C}},\,\fml{F}\setminus\fml{C}'$ is
  not all-falsifiable.
  FMCFS is the function problem of computing an MCFS of $\fml{F}$.
\end{definition}

\begin{remark} \label{rm:mfs}
  $\fml{M}$ is an MFS of $\fml{F}$ iff
  $\fml{N}=\fml{F}\setminus\fml{D}$ is an MCFS of of $\fml{F}$.
  %
  %
\end{remark}

\begin{definition}[FLMFS/\mxf] \label{def:lmfs}
  Let $\fml{F}\in\CNF$. An MFS $\fml{M}$ of $\fml{F}$ is a largest MFS
  (LMFS) iff for any MFS $\fml{M}'$  of $\fml{F}$,
  $|\fml{M}'|\le|\fml{M}|$.
  FLMFS (or {\em Maximum Falsifiability}, \mxf) is the function
  problem of computing an LMFS of $\fml{F}$.
\end{definition}

\begin{definition}[FSMCFS/\mnsat] \label{def:smcfs}
  Let $\fml{F}\in\CNF$. An MCFS $\fml{C}$ of $\fml{F}$ is a smallest
  MFS (SMCFS) iff for any MCFS $\fml{C}'$ of $\fml{F}$,
  $|\fml{C}'|\ge|\fml{C}|$. 
  FLMCFS (or {\em Minimum Satisfiability}, \mnsat) is the function
  problem of computing an SMCFS of $\fml{F}$.
  %
\end{definition}

\begin{remark}
  Clearly, $\fml{M}\subseteq\fml{F}$ is an LMFS of $\fml{F}$ iff
  $\fml{C} = \fml{F}\setminus\fml{M}$ is an SMCFS of $\fml{F}$.
\end{remark}

\begin{remark}
  Given \autoref{rm:mfs}, one can conclude that the \mnsat
  problem consists of computing the smallest MCFS, and this represents
  the complement of the \mxf solution.
  %
  Thus, for $\fml{F}\in\CNF$, where $n_t$ and $n_f$ denote
  respectively the \mnsat and the \mxf solutions, then $|\fml{F}| = 
  n_t + n_f$.
\end{remark}

As with minimal unsatisfiability and irredundancy, maximal
falsifiability can be generalized to the case when some clauses are
hard and when soft clauses have weights.
Similarly, one could consider the DNF formulas, and function problems
related with \emph{all-true} terms.

The \mnsat problem has been studied
in~\cite{kohli-sjdm94,ravi-ipl96,zwick-tcs05,cli-aij13}. The problem
of maximal falsifiability is studied in~\cite{impms-lpar13}.

Moreover, it should be noted that, although this paper addresses
MUSes/\-MCSes/\-MSSes/\-MESes/\-MDSes/\-MNSes/\-MFSes/\-MCFSes
defined solely on sets of clauses, other variants could be considered,
namely groups of clauses or variables.
The formalizations of these variants as instances of the MSMP problem
mimic the formalizations developed for the case of sets of clauses.
(See~\cite{liffiton-jar08,hertz-jco09,nadel-fmcad10,hmms-ccai12,bimms-sat12,hmms-aicomm13}
for related work on groups of clauses and variables.)

\subsubsection{Minimal \& Maximal Models} \label{ssec:mms}

Minimal and maximal models are additional examples of function
problems associated with propositional formulas.

\begin{definition}[Minimal Model; FMnM] \label{def:mnm}
  Given $\fml{F}\in\BFML$, a model $\fml{M}\subseteq X$ of $\fml{F}$ is
  minimal iff $\nu(\fml{M},X)\entails\fml{F}$ and 
  $\forall_{\fml{M}'\subsetneq\fml{M}},\,\nu(\fml{M}',X)\nentails\fml{F}$.
  FMnM is the function problem of computing a minimal model of $\fml{F}$.
\end{definition}

\begin{definition}[Maximal Model; FMxM] \label{def:lmxm}
  Given $\fml{F}\in\BFML$,
  a model $\fml{M}\subseteq X$ of $\fml{F}$ is maximal iff
  $\nu(\fml{M},X)\entails\fml{F}$ and
  $\forall_{\fml{F}\supseteq\fml{M}'\supsetneq\fml{M}},\,\nu(\fml{M}',X)\nentails\fml{F}$.
  FMxM is the function problem of computing a maximal model of $\fml{F}$.
\end{definition}

\begin{observation}
  Both minimal and maximal models can be computed subject to a set $Z$
  of variables other than
  $X\triangleq\var(\fml{F})$~\cite{ben-eliyahu-amai96}, i.e.\ the
  so-called $Z$-minimal and $Z$-maximal models. The above definitions
  can easily be modified for this more general definition.
\end{observation}

Minimal models find a wide range of applications, including
non-monotonic reasoning and bioinformatics
(e.g.\ \cite{eiter-tcs93,inoue-ecai10}).
Algorithms for computing minimal and maximal models have been studied
in the past
(e.g.\ \cite{ben-eliyahu-amai96,niemela-tableaux96,ben-eliyahu-aij97,kavvadias-ipl00,hasegawa-ftp09}).
In some contexts, minimal and maximal models have been referred to as
MIN-ONE$_\subseteq$ and MAX-ONE$_{\subseteq}$ solutions
(e.g.\ \cite{giunchiglia-cj10}).

\subsubsection{Implicants \& Implicates} \label{ssec:piip}

Two relevant function problems associated with propositional formulas
consist of computing prime implicants, starting from an implicant
represented as a term, and prime implicates, starting from an
implicate represented as a clause.

\begin{definition}[Prime Implicant (given term); FPIt]\label{def:pit}
  Given $\fml{F}\in\BFML$ and an implicant $t\in\ASGN$ of $\fml{F}$,
  term $u\subseteq t$ is a prime implicant of $\fml{F}$ iff
  $u\entails\fml{F}$ and $\forall_{v\subsetneq u},v\nentails\fml{F}$.
  FPIt is the function problem of computing a prime implicant of
  $\fml{F}$ given an implicant $t$ of $\fml{F}$.
\end{definition}

\begin{definition}[Prime Implicate (given clause); FPIc] \label{def:pic}
  Given $\fml{F}\in\BFML$ and an implicate $c\in\ASGN$ of $\fml{F}$,
  clause $p\subseteq c$ is a prime implicate of $\fml{F}$ iff  
  $\fml{F}\entails p$ and $\forall_{q\subsetneq p},\,\fml{F}\nentails q$.
  FPIc is the function problem of computing a prime implicate of
  $\fml{F}$ given an implicate $c$ of $\fml{F}$.
\end{definition}

Observe that computing a prime implicant from a given implicant for a
formula in CNF can be done in linear time
(e.g.\ \cite{somenzi-tacas04,leberre-fmcad13}). Similarly, computing a 
prime implicate from a given implicate for a formula in DNF can also
be done in linear time.

For the general case of arbitrary Boolean formulas these function
problems become significantly harder, and the algorithms described in
this paper require a number of SAT solver calls that is linear in the
number of variables in the worst case
(see~\cite{toda-ic95,umans-tcad06} for similar approaches and
conclusions). 


Prime implicates and implicants find many practical applications,
including 
truth maintenance systems~\cite{kleer-aaai87,kleer-aaai92},
knowledge compilation~\cite{darwiche-jair02},
conformant planning~\cite{pontelli-aaai11},
the simplification of Boolean
functions~\cite{quine-amm55,mccluskey-bltj56},
abstraction in model checking~\cite{bradley-fac08,bradley-fmcad07},
minimization of counterexamples in model
checking~\cite{somenzi-tacas04,shen-vmcai05,shen-date05},
among many others.
A wealth of algorithms exist for computing prime implicates and
implicants (e.g.\ \cite{marquis-hdrums00} and references therein).

Another problem of interest in practice is the Longest Extension of
Implicant problem~\cite{umans-tcad06} for DNF formulas. By analogy, we
also consider the Longest Extension of Implicate problem for CNF
formulas.

\begin{definition}[Longest Extension of Implicant; FLEIt] \label{def:leit}
  Let $t$ be an implicant of $\fml{F}\in\DNF$. Term $\ASGN\supseteq
  u\supseteq t$ is a \emph{longest extension of implicant} $t$ iff
  $(\fml{F}\setminus\{t\})\cup\{u\}\equiv\fml{F}$ and
  $\forall_{u'\supsetneq u},(\fml{F}\setminus\{t\})\cup\{u'\}\not\equiv\fml{F}$. 
  The FLEIt function problem consists of computing a longest
  extension of implicant $t$ of $\fml{F}$.
\end{definition}

\begin{definition}[Longest Extension of Implicate; FLEIc] \label{def:leic}
  Let $c$ be an implicate of $\fml{F}\in\CNF$. Clause $\ASGN\supseteq
  u\supseteq c$ is a \emph{longest extension of implicate} $c$ iff
  $(\fml{F}\setminus\{c\})\cup\{u\}\equiv\fml{F}$ and
  $\forall_{u'\supsetneq u},
  (\fml{F}\setminus\{c\})\cup\{u\}\not\equiv\fml{F}$. 
  The FLEIc function problem consists of computing a longest
  extension of implicate $c$ of $\fml{F}$.
\end{definition}

\subsubsection{Formula Entailment}

A number of problems related with formula entailment can also be
considered. We address two function problems: computing a minimal set
entailing a formula and computing a maximal set entailed by a
formula.

\begin{definition}[Minimal Entailing Subset; FMnES] \label{def:lmnes}
  Let $\fml{J}\in\CNF$ and $\fml{I}\in\BFML$ be such that
  $\fml{J}\entails\fml{I}$. $\fml{M}\subseteq\fml{J}$ is a {\em
    Minimal Entailing Subset} (MnES) of $\fml{I}$ iff
  $\fml{M}\entails\fml{I}$ and
  $\forall_{\fml{M}'\subsetneq\fml{M}}\,\fml{M}'\nentails\fml{I}$.
  FMnES is the function problem of computing a minimal entailing
  subset of $\fml{J}$ given $\fml{I}$.
\end{definition}

\begin{observation}
  The function problem FMnES from~\autoref{def:lmnes} can be
  generalized to capture the case of logic-based
  abduction~\cite{gottlob-jacm95} when the universe of hypotheses is
  consistent with the theory, and the minimality criterion is subset
  minimality.
\end{observation}

\begin{definition}[Maximal Entailed Subset; FMxES] \label{def:mxes}
  Let $\fml{J}\in\BFML$ and $\fml{N}\in\CNF$ be such that
  $\fml{J}\nentails\fml{N}$. $\fml{I}\subseteq\fml{N}$ is a {\em
    Maximal Entailed Subset} (MxES) iff $\fml{J}\entails\fml{I}$ and
  $\forall_{\fml{N}\supseteq\fml{I}'\supsetneq\fml{I}}\,\fml{J}\nentails\fml{I}'$.
  FMxES is the function problem of computing a maximal entailed
  subset of $\fml{N}$ given $\fml{J}$.
\end{definition}

\subsubsection{Backbone Literals}

The problem of computing the backbone of a Boolean formula, i.e.\ the
literals that are common to all satisfying assignments of the formula,
finds a wide range of applications. Two definitions are considered.

\begin{definition}[Backbone; FBB] \label{def:bb}
  Given $\fml{F}\in\BFML$, the \emph{backbone} of $\fml{F}$ is a
  maximal set of literals $\fml{B}\in\ASGN$ which are true in all
  models of $\fml{F}$, i.e.\ $\fml{F}\entails\land_{l\in\fml{B}}(l)$.
  FBB is the function problem of computing the backbone of
  $\fml{F}\in\BFML$.
\end{definition}

In practice, algorithms for computing the backbone of a formula often
start from a reference model~\cite{msjl-ecai10}. This allows a
slightly different formulation of the backbone function problem.

\begin{definition}[FBBr] \label{def:bbr}
  Let $\nu$ be a model of $\fml{F}\in\BFML$ and $\fml{V}$ the
  associated set of literals.
  The \emph{backbone} of $\fml{F}$ is a
  maximal set of literals $\fml{B}\subseteq\fml{V}$ such that
  $\fml{F}\entails\land_{l\in\fml{B}}(l)$.
  FBBr is the function problem of computing the backbone of
  $\fml{F}\in\BFML$ given $\fml{V}$.
\end{definition}

Backbones of Boolean formulas were first studied in the context of
work on phase transitions and problem
hardness~\cite{monasson-rsa99,walsh-ijcai01,walsh-aaai05a}. In
addition, backbones find several practical applications, that include
configuration~\cite{sinz-flairs01,sinz-aiedam03},
abstract argumentation~\cite{weissenbacher-clima13},
cyclic causal models~\cite{jarvisalo-uai13},
debugging of fabricated circuits~\cite{malik-fmcad11}, among others.
A number of recent algorithms have been proposed for computing the
backbone of Boolean
formulas~\cite{kuechlin-ijcar01,janota-splc08,msjl-ecai10,malik-hldvt11}.
Recently, generalized backbones were studied in~\cite{codish-hvc13}.

\subsubsection{Variable Independence}  

A formula $\fml{F}$ is (semantically) independent of a variable $x$ if
the set of models of $\fml{F}$ does not change by fixing $x$ to {\em
  any} truth value~\cite[Def.\ 4, Prop.\ 7]{lang-jair03}.

\begin{definition}[Variable Independence; FVInd] \label{def:vind}
  A $\fml{F}\in\BFML$ is independent from $x\in\var(\fml{F})$
  iff $\fml{F}\equiv\fml{F}_{x=0}$ (or, $\fml{F}\equiv\fml{F}_{x=1}$).
  FVInd is the function problem of computing a maximal set of
  variables of which $\fml{F}$ is independent from.
\end{definition}

Observe that, as indicated earlier, the definition of MESes can be
extended to variables. However, FVInd is defined in a more general
setting, since $\fml{F}$ need not be in CNF.
Variable independence (and also literal independence) are important in
a number of settings~\cite{lang-jair03}.
Variables declared independent are also known as {\em redundant} or
{\em inessential}, e.g.~\cite{dietmeyer-tec67,brown-bk90,crama-bk11}.
In later sections, and for simplicitly, if $\fml{F}$ is independent
from $x$, then we write that $x$ is redundant for $\fml{F}$.

\subsubsection{Maximum Autarkies}

This section provides a brief overview of the problem of identifying
the (maximum) autarkies of unsatisfiable CNF formulas.

\begin{definition}[Autarky; FAut] \label{def:aut}
  Given $\fml{F}\in\CNF$, with $\fml{F}\entails\bot$, a set
  $\fml{A}\in\var(\fml{F})$ is an autarky iff there exists a truth 
  assignment to the variables in $\fml{A}$ that satisfies all clauses
  containing literals in the variables of $\fml{A}$.
  FAut is the function problem of computing the maximal autarky of
  $\fml{F}$. 
\end{definition}

Autarkies were first proposed in the context of improving exponential 
upper bounds of SAT
algorithms~\cite{speckenmeyer-dam85,vangelder-dam99}, and have later
been studied in the context of minimal unsatisfiable
subformulas~\cite{kullmann-dam00,kullmann-hdbk09a,kullmann-fi11a}. More 
recently, autarkies were used  in model
elimination~\cite{vangelder-jar99} and for speeding up MCS/MUS
enumeration algorithms~\cite{liffiton-sat08}. Algorithms for computing
autarkies were studied
in~\cite{kullmann-endm01,klms-sat06,liffiton-sat08,kullmann-hdbk09a}.

\subsection{Properties} \label{ssec:props}

This section summarizes properties of some of the function problems
presented  in the previous section, which are essential for some of
the results presented later.
%
%

For some function problems, the number of maximal/minimal sets is very
restricted. In fact, for FLEIt, FLEIc, FMxES, FBBr, FBB and FAut, the
following holds.
\begin{proposition} \label{prop:unique}
  For the function problems FLEIt, FLEIc, FMxES, FBB, FBBr, and FAut
  there is a unique maximal set, which is maximum.
\end{proposition}

\begin{proof}
  For each case, the proof is by contradiction.
  \begin{enumerate}
  \item For FLEIt, let $u$ be a subset-maximal extension of $t$ such
    that $u\entails\fml{F}$, and let $l_x$ be a literal not included
    in $u$ such that $t\land l_x\entails\fml{F}$. Then, $u\land
    l_x\entails\fml{F}$; a contradiction.
  \item For FLEIc the proof is similar to the previous case.
  \item For FMxES, let $\fml{I}$ be a maximal set such that
    $\fml{J}\entails\fml{I}$, and assume there exists clause
    $c\in\fml{N}\setminus\fml{I}$ such that $\fml{J}\entails c$.
    Then, any model of $\fml{J}$ satisfies both $\fml{I}$ and $c$, and
    so $\fml{J}\entails\fml{I}\cup\{c\}$; a contradiction.
  \item For FBB/FBBr, the proof is similar. Let $\fml{B}$ denote a
    maximal set of backbone literals, and let $l\not\in\fml{B}$ be a
    backbone literal. Then, for any model of the formula, all literals
    in $\fml{B}\cup\{l\}$ are true; a contradiction.
  \item For FAut, the proof is again similar. Let $\fml{A}_1$ be a
    maximal autarky, let $\fml{A}_2$ be another autarky, and let
    $\fml{A}_2\setminus\fml{A}_1\not=\emptyset$. Then
    $\fml{K} = \fml{A}_1\cup\fml{A}_2$ is an autarky and
    $\fml{A}_1\subseteq\fml{K}$; a contradiction. 
  \end{enumerate}
  Thus, for FLEIt, FLEIc, FMxES, FBB, FBBr and FAut there is a unique
  maximal set which is maximum.
\end{proof}

\begin{proposition}
  A formula $\fml{F}$ is independent from $x_i\in X$ iff
  $\fml{F}\equiv\fml{F}[x_i/y_i]$, where $y_i$ is a new variable with
  $y_i\not\in X$. 
\end{proposition}

\begin{proof}
  By definition, $\fml{F}$ is independent from $x_i$ iff
  $\fml{F}\equiv\fml{F}_{x=0}$ or $\fml{F}\equiv\fml{F}_{x=1}$.
  Hence, $\forall_{y_i\in\{0,1\}}\fml{F}\equiv\fml{F}[x_i/y_i]$,
  with $y_i\not\in X$. Thus, $\fml{F}\equiv\fml{F}[x_i/y_i]$ with
  $y_i\not\in X$.
\end{proof}


%
%
%

\section{Reductions to MSMP} \label{sec:maps}

This section shows how each of the function problems defined in
\autoref{sec:probs} can be represented as an instance of the MSMP
problem. 
%
%
%
\autoref{ssec:forms} introduces general predicate forms, which are
shown to be monotone, and which simplify the presentation of the
reductions in the \autoref{ssec:greds}.
\autoref{ssec:greds} is structured similarly to \autoref{ssec:fdefs}:
(i) minimal unsatisfiability (and maximal satisfiability);
(ii) irredundant subformulas;
(iii) maximal falsifiability (and minimal satisfiability);
(iv) minimal and maximal models;
(v) prime implicates and implicants;
(vi) backbone literals;
(vii) formula entailment;
(viii) variable independence; and
(ix) maximum autarkies.

\subsection{Predicate Forms} \label{ssec:forms}

In the next section, several function problems are reduced to the
MSMP. The reduction involves specifying a reference set $\fml{R}$ and
a monotone predicate $\Pred$ defined in terms of a working set
$\fml{W}\subseteq\fml{R}$.
To simplify the description of the different monotone predicates, this
section develops general predicate forms, which capture all of the
monotone predicates developed in the next sections, and proves that
all predicates of any of these forms are monotone. As a result, for
any concrete predicate, monotonicity is an immediate consequence of
the monotonicity of the general predicate forms.

Let element $u_i\in\fml{R}$ represent either a literal or a clause.
Moreover, $\sigma(u_i)$ represents a Boolean formula built from $u_i$,
where new variables may be used, but such that $u_i$ is the only
element from $\fml{R}$ used in $\sigma(u_i)$. For example,
$\sigma(u_i)$ can represent the negation of a literal or a clause,
etc. 
Let $\fml{G}$ be a propositional formula that is {\em independent}
from the elements in $\fml{W}$, i.e.\ $\fml{G}$ does not change with
$\fml{W}$.
Then, the following general predicate forms are defined.
\begin{definition}[Predicates of Form \lform]
  A predicate is of \emph{form $\mathscr{L}$} iff its general form is
  given by,
  \begin{equation} \label{eq:reflogpred}
    \Pred(\fml{W})\triangleq\SAT(\fml{G}\land\land_{u_i\in\fml{R}\setminus\fml{W}}\,(\sigma(u_i)))
  \end{equation}
\end{definition}
\begin{definition}[Predicates of Form \pform]
  A predicate is of \emph{form $\mathscr{P}$} iff its general form is
  given by,
  \begin{equation} \label{eq:refnppred}
    \Pred(\fml{W})\triangleq\neg\SAT(\fml{G}\land\land_{u_i\in\fml{W}}\,(\sigma(u_i)))
  \end{equation}
\end{definition}
\begin{definition}[Predicates of Form \bform]
  A predicate is of \emph{form $\mathscr{B}$} iff its general form is
  given by,
  \begin{equation} \label{eq:refparpred}
    \Pred(\fml{W})\triangleq\neg\SAT(\fml{G}\land(\lor_{u_i\in\fml{R}\setminus\fml{W}}\,(\sigma(u_i))))
  \end{equation}
\end{definition}

\begin{proposition} \label{prop:mono}
  Predicates of the forms \lform, \bform and \pform are monotone.
\end{proposition}

\begin{proof}
  \begin{enumerate}
  \item Let $\Pred$ be a predicate of form \lform.
    Let $\Pred(\fml{R}_0)$ hold, with $\fml{R}_0\subseteq\fml{R}$,
    i.e.\ the argument to the SAT oracle is satisfiable.
    Then, by \autoref{prop:satsubset}, for any $\fml{R}_1$, with
    $\fml{R}_1\supseteq\fml{R}_0$ (and so
    $\fml{R}\setminus\fml{R}_1\subseteq\fml{R}\setminus\fml{R}_0$),
    $\Pred(\fml{R}_1)$ also holds, since the argument to the SAT
    oracle call is the conjunction of a subset of the constraints used
    for the case of $\fml{R}_0$, and so also satisfiable.
    Thus, $\Pred$ is monotone.
  \item Let $\Pred$ be a predicate of form \pform.
    Let $\Pred(\fml{R}_0)$ hold, with $\fml{R}_0\subseteq\fml{R}$,
    i.e.\ the argument to the SAT oracle is unsatisfiable.
    Then, by \autoref{prop:usatsupset}, for any $\fml{R}_1$, with
    $\fml{R}_1\supseteq\fml{R}_0$, $\Pred(\fml{R}_1)$ also holds,
    since the argument to the SAT oracle call is the conjunction of a
    superset of the constraints used for the case of $\fml{R}_0$, and
    so also unsatisfiable.
    Thus, $\Pred$ is monotone.
  \item Let $\Pred$ be a predicate of form \bform.
    Let $\Pred(\fml{R}_0)$ hold, with $\fml{R}_0\subseteq\fml{R}$,
    i.e.\ the argument to the SAT oracle is unsatisfiable.
    Then, for any $\fml{R}_1$, with $\fml{R}_1\supseteq\fml{R}_0$,
    $\Pred(\fml{R}_1)$ also holds, since the clause created from
    $\fml{R}\setminus\fml{R}_1$ has fewer elements than for the case of
    $\fml{R}\setminus\fml{R}_0$, and so it is also unsatisfiable
    (i.e.\ all literals in the clause will also be resolved away).
    %
    Thus, $\Pred$ is monotone.
  \end{enumerate}
\end{proof}

\subsection{Minimal Sets} \label{ssec:greds}

In the remainder of this section the reference set is denoted
$\fml{R}$ and the monotone predicate $\Pred$ is defined in terms of a
working set $\fml{W}\subseteq\fml{R}$.
All function problems considered are defined in \autoref{sec:probs}.

\subsubsection{Minimal Unsatisfiability \& Maximal Satisfiability}

\begin{proposition} \label{prop:fmus}
  $\tn{FMUS}\maps\tn{MSMP}$.
\end{proposition}

\begin{proof}
\proofitem{Reduction}
The reduction is defined as follows. $\fml{R}\triangleq\fml{F}$ and,
\begin{equation} \label{eq:muspred}
  \Pred(\fml{W})\triangleq\neg\SAT(\land_{c\in\fml{W}}\,(c))
\end{equation}
with $\fml{W}\subseteq\fml{R}$.
\proofitem{Monotonicity}
The predicate (see \eqref{eq:muspred}) is of form \pform, with
$\fml{G}\triangleq\emptyset$, $u_i\triangleq c$, and
$\sigma(c)\triangleq c$. Thus, by \autoref{prop:mono} the predicate is
monotone.
%
%
\proofitem{Correctness}
Let $\fml{M}$ be a minimal set for which $\Pred(\fml{M})$ holds,
i.e.\ $\fml{M}$ is unsatisfiable. By \autoref{def:mp}, since
$\fml{M}$ is minimal for $\Pred$, then for any
$\fml{M}'\subsetneq\fml{M}$ the predicate does not hold,
i.e.\ $\fml{M}'$ is satisfiable.
Thus, by \autoref{def:mus}, $\fml{M}$ is an MUS of $\fml{F}$.
\end{proof}

\begin{proposition} \label{prop:fmcs}
  $\tn{FMCS}\maps\tn{MSMP}$.
\end{proposition}

\begin{proof}
\proofitem{Reduction}
The reduction is defined as follows. $\fml{R}\triangleq\fml{F}$ and,
\begin{equation} \label{eq:mcspred}
  \Pred(\fml{W})\triangleq\SAT(\land_{c\in\fml{R}\setminus\fml{W}}\,(c))
\end{equation}
with $\fml{W}\subseteq\fml{R}$.
\proofitem{Monotonicity}
The predicate (see \eqref{eq:mcspred}) is of form \lform, with
$\fml{G}\triangleq\emptyset$, $u_i\triangleq c$, and
$\sigma(c)\triangleq c$. Thus, by \autoref{prop:mono} the predicate is
monotone.
%
%
\proofitem{Correctness}
Let $\fml{M}$ be a minimal set such that $\Pred(\fml{M})$ holds.
By \autoref{def:mp}, since $\fml{M}$ is minimal for $\Pred$, then for
any $\fml{M}'\subsetneq\fml{M}$ the predicate does not hold,
i.e.\ $\fml{R}\setminus\fml{M}'$ is unsatisfiable.
Hence, $\fml{F}\setminus\fml{M}$ is satisfiable, and for any
$\fml{M}'\subsetneq\fml{M}$, $\fml{F}\setminus\fml{M}'$ is
unsatisfiable.
Thus, by \autoref{def:mcs}, $\fml{M}$ is an MCS of $\fml{F}$. 
\end{proof}

\begin{remark} \label{rm:fmss}
  By \autoref{rm:mss}, an MSS of $\fml{F}\in\CNF$ can be computed as
  follows. Compute an MCS $\fml{M}$ of $\fml{F}$ and return
  $\fml{F}\setminus\fml{M}$.
  %
  %
\end{remark}

\subsubsection{Irredundant Subformulas}

\begin{proposition} \label{prop:fmes}
  $\tn{FMES}\maps\tn{MSMP}$.
\end{proposition}

\begin{proof}
\proofitem{Reduction}
The reduction is defined as follows. $\fml{R}\triangleq\fml{F}$ and,
\begin{equation} \label{eq:mespred}
  \Pred(\fml{W})\triangleq\neg\SAT(\neg\fml{F}\land\land_{c\in\fml{W}}\,(c))
\end{equation}
with $\fml{W}\subseteq\fml{R}$.
\proofitem{Monotonicity}
The predicate (see \eqref{eq:mespred}) is of form \pform, with
$\fml{G}\triangleq\neg\fml{F}$, $u_i = c$, and $\sigma(c)\triangleq
c$. Thus, by \autoref{prop:mono} the predicate is monotone.
%
%
\proofitem{Correctness}
Let $\fml{M}$ be a minimal set such that $\Pred(\fml{M})$ holds.
By \autoref{def:mp}, since $\fml{M}$ is minimal for $\Pred$, then
$\Pred$ holds for $\fml{M}$, i.e.\ $\fml{M}\entails\fml{F}$, and for
any $\fml{M}'\subsetneq\fml{M}$ the predicate does not hold,
i.e.\ $\fml{M}'\nentails\fml{F}$. 
Thus, by \autoref{def:mes}, $\fml{M}$ is an MES of $\fml{F}$.
\end{proof}

\begin{observation} \label{obs:gmus}
  Observe that the problem of computing an irredundant subformula can
  be reduced to the problem of computing a (group)
  MUS~\cite{bjlms-cp12}. Thus, an MUS for the resulting problem is an
  MES for the original problem.
\end{observation}

Also, note that the argument of the predicate in~\eqref{eq:mespred} can
be simplified:
\begin{equation*}
  \begin{array}{lcl}
    \neg\fml{F}\land\land_{c\in\fml{W}}\,(c)
    & \Leftrightarrow & 
    (\lor_{c\in\fml{F}}(\neg c))\land\land_{c\in\fml{W}}(c) 
    \Leftrightarrow (\lor_{c\in\fml{F}\setminus\fml{W}}(\neg c))\land\land_{c\in\fml{W}}(c) \\
    & \Leftrightarrow &
    \neg(\fml{F}\setminus\fml{W})\land\land_{c\in\fml{W}}\,(c) \\
  \end{array}
\end{equation*}
Thus, the predicate in~\eqref{eq:mespred} can be formulated as
follows:
\begin{equation} \label{eq:mespred2}
  \Pred(\fml{W})\triangleq\neg\SAT(\neg(\fml{F}\setminus\fml{W})\land\land_{c\in\fml{W}}\,(c))
\end{equation}

Throughout the paper,~\eqref{eq:mespred} is used, since it facilitates
relating this predicate with others. However, for practical
purposes~\eqref{eq:mespred2} would be preferred.

\begin{proposition} \label{prop:fmds}
  $\tn{FMDS}\maps\tn{MSMP}$.
\end{proposition}

\begin{proof}
\proofitem{Reduction}
The reduction is defined as follows.
$\fml{R}\triangleq\fml{F}$ and,
\begin{equation} \label{eq:mdspred}
  \Pred(\fml{W})\triangleq\SAT(\neg\fml{F}\land\land_{c\in\fml{R}\setminus\fml{W}}\,(c))
\end{equation}
with $\fml{W}\subseteq\fml{R}$.
\proofitem{Monotonicity}
The predicate (see \eqref{eq:mdspred}) is of form \lform, with
$\fml{G}\triangleq\neg\fml{F}$, $u_i\triangleq c$, and
$\sigma(c)\triangleq c$. Thus, by \autoref{prop:mono} the predicate is
monotone.
\proofitem{Correctness}
Let $\fml{D}$ be a minimal set such that $\Pred(\fml{D})$ holds.
By \autoref{def:mp}, since $\fml{D}$ is minimal for $\Pred$, then
$\Pred$ holds for $\fml{D}$,
i.e.\ $\fml{F}\setminus\fml{D}\nentails\fml{F}$, and for any
$\fml{D}'\subsetneq\fml{D}$ the predicate does not hold,
i.e.\ $\fml{F}\setminus\fml{D}'\entails\fml{F}$.
Thus, by \autoref{def:mds}, $\fml{D}$ is an MDS of $\fml{F}$.
\end{proof}

\begin{observation}
  The problem of computing an irredundant subformula can be reduced to
  the problem of computing a (group) MUS (see \autoref{obs:gmus}
  and~\cite{bjlms-cp12}), and so an MCS for the resulting problem is
  an MDS for the original problem.
\end{observation}

Similarly to the FMES case, the argument to the predicate
in~\eqref{eq:mdspred} can be simplified:
\begin{equation*} 
  \begin{array}{lcl}
    \neg\fml{F}\land\land_{c\in\fml{R}\setminus\fml{W}}\,(c)
    & \Leftrightarrow &
    (\lor_{c\in\fml{F}}(\neg c))\land\land_{c\in\fml{F}\setminus\fml{W}}(c)
    \Leftrightarrow (\lor_{c\in\fml{W}}(\neg c))\land\land_{c\in\fml{F}\setminus\fml{W}}(c) \\
    & \Leftrightarrow &
    \neg\fml{W}\land\land_{c\in\fml{R}\setminus\fml{W}}\,(c) \\
  \end{array}
\end{equation*}
Thus, the predicate in~\eqref{eq:mdspred} can be formulated as
follows:
\begin{equation} \label{eq:mdspred2}
  \Pred(\fml{W})\triangleq\SAT(\neg\fml{W}\land\land_{c\in\fml{R}\setminus\fml{W}}\,(c))
\end{equation}

Throughout the paper,~\eqref{eq:mdspred} is used, since it facilitates 
relating this predicate with others. However, for practical
purposes~\eqref{eq:mdspred2} would be preferred.

\begin{remark} \label{rm:fmns}
  By \autoref{rm:mns}, an MNS of $\fml{F}\in\CNF$ can be computed as
  follows. Compute an MDS $\fml{D}$ of $\fml{F}$ and return
  $\fml{F}\setminus\fml{D}$.
\end{remark}

\subsubsection{Minimal Satisfiability \& Maximal Falsifiability}

\begin{proposition} \label{prop:fmcfs}
  $\tn{FMCFS}\maps\tn{MSMP}$.
\end{proposition}

\begin{proof}
\proofitem{Reduction}
The reduction is defined as follows. $\fml{R}\triangleq\fml{F}$ and,
\begin{equation} \label{eq:mfspred}
  \Pred(\fml{W})\triangleq\SAT(\land_{c\in\fml{R}\setminus\fml{W}}\,(\neg c))
\end{equation}
with $\fml{W}\subseteq\fml{R}$.
\proofitem{Monotonicity}
The predicate (see \eqref{eq:mfspred}) is of form \lform, with
$\fml{G}\triangleq\emptyset$, $u_i\triangleq c$, and
$\sigma(c)\triangleq\neg c$. Thus, by \autoref{prop:mono} the
predicate is monotone.
\proofitem{Correctness}
Let $\fml{M}$ be a minimal set such that $\Pred(\fml{M})$ holds.
By \autoref{def:mp}, since $\fml{M}$ is minimal for $\Pred$, then for
any $\fml{M}'\subsetneq\fml{M}$ the predicate does not hold,
i.e.\ $\fml{F}\setminus\fml{M}$ is all-falsifiable, and for any
$\fml{M}'\subsetneq\fml{M}$, $\fml{F}\setminus\fml{M}'$ is not
all-falsifiable. Thus, by \autoref{def:mfs}, $\fml{M}$ is an MCFS of
$\fml{F}$.
\end{proof}


\begin{remark} \label{rm:fmcfs}
  By \autoref{rm:mfs}, an MFS of $\fml{F}\in\CNF$ can be computed as
  follows. Compute an MCFS $\fml{C}$ of $\fml{F}$ and return
  $\fml{F}\setminus\fml{C}$.
\end{remark}

\subsubsection{Minimal \& Maximal Models}

\begin{proposition} \label{prop:fmnm}
  $\tn{FMnM}\maps\tn{MSMP}$.
\end{proposition}

\begin{proof}
\proofitem{Reduction}
The reduction is defined as follows. $\fml{R}\triangleq
X\triangleq\var(\fml{F})$
and,
\begin{equation} \label{eq:mnmpred}
  \Pred(\fml{W})\triangleq\SAT(\fml{F}\land\land_{x\in\fml{R}\setminus\fml{W}}\,(\neg x))
\end{equation}
with $\fml{W}\subseteq\fml{R}$.
\proofitem{Monotonicity}
The predicate (see \eqref{eq:mnmpred}) is of form \lform, with
$\fml{G}\triangleq\fml{F}$, $u_i\triangleq x$, and
$\sigma(x)\triangleq\neg x$. Thus, by \autoref{prop:mono} the
predicate is monotone.
%
%
\proofitem{Correctness}
Let $\fml{M}$ be a minimal set such that $\Pred(\fml{M})$ holds.
By \autoref{def:mp}, since $\fml{M}$ is minimal for $\Pred$, then for
any $\fml{M}'\subsetneq\fml{M}$ the predicate does not
hold.
By observing that $\fml{M}$ denotes the set of variables that {\em
  can} be assigned value 1 (since the other variables {\em must} be
assigned value 0) and, because $\fml{M}$ is minimal, $\fml{M}$ cannot
be further reduced. Thus, by \autoref{def:mnm}, $\fml{M}$ is a minimal
model of $\fml{F}$.
%
\end{proof}

\begin{proposition} \label{prop:fmxm}
  $\tn{FMxM}\maps\tn{MSMP}$.
\end{proposition}

\begin{proof}
  Let $\fml{F}^C$ be constructed from $\fml{F}$ by flipping the
  polarity of all literals in $\fml{F}$, i.e.\ replace $l$ with $\neg
  l$ for $l\in\{x,\neg x\,|\,x\in\var(\fml{F})\}$. Observe that,
  $\fml{F}$ and $\fml{F}^C$ have the same parse tree excluding the
  leaves.
  Now, compute a minimal model $M^C$ for $\fml{F}^C$, e.g.\ using
  \autoref{prop:fmnm}. Then, as shown next, a maximal model for
  $\fml{F}$ is given by $\var(\fml{F})\setminus M^C$.

  Let $M^C$ be any model of $\fml{F}^C$, and let $\nu^C(\fml{M}^C,X)$
  denote the associated truth assignment. Consider the truth
  assignment $\nu$ obtained by flipping the value of all variables,
  and let $\fml{M}$ denote the variables assigned value 1,
  i.e.\ $\fml{M} = X\setminus\fml{M}^C$. Clearly, $\fml{M}^C$ is
  minimal iff $\fml{M}$ is maximal.
  Moreover, let $\fml{F}$ be obtained from $\fml{F}^C$ by
  complementing all of its literals.
  Thus, the leaves of the parse tree of $\fml{F}^C$ are complemented
  and the values assigned to the leaves are also complemented.
  %
  %
  Now, recall that, with the exception of the leaves, both $\fml{F}$
  and $\fml{F}^C$ have the same parse tree, and the leaves are
  assigned the same values in both cases. Thus, by structural 
  induction it follows that $\nu^C\entails{\fml{F}^C}$ iff
  $\nu\entails\fml{F}$. 
  %
  %
\end{proof}

\subsubsection{Implicants \& Implicates}

\begin{proposition} \label{prop:fpit}
  $\tn{FPIt}\maps\tn{MSMP}$.
\end{proposition}

\begin{proof}
\proofitem{Reduction}
The reduction is defined as
follows. $\fml{R}\triangleq L(t)\triangleq\{l\,|\,l\in t\}$ and,
\begin{equation} \label{eq:pitpred}
  \Pred(\fml{W})\triangleq\neg\SAT(\neg\fml{F}\land\land_{l\in\fml{W}}\,(l))
\end{equation}
with $\fml{W}\subseteq\fml{R}$.
\proofitem{Monotonicity}
The predicate (see \eqref{eq:pitpred}) is of form \pform, with
$\fml{G}\triangleq\neg\fml{F}$, $u_i\triangleq l$, and
$\sigma(l)\triangleq l$. Thus, by \autoref{prop:mono} the predicate is
monotone.
%
%
\proofitem{Correctness}
Let $\fml{M}$ be a minimal set such that $\Pred(\fml{M})$ holds.
Thus, the literals in $\fml{M}$ entail $\fml{F}$, and for any proper
subset $\fml{M}'$ of $\fml{M}$, the literals in $\fml{M}'$ do not
entail $\fml{F}$. Thus $\fml{M}$ is a prime implicant of $\fml{F}$.
\end{proof}

Similarly, we can reduce the computation of a prime implicate given a
clause to MSMP. (This reduction of FPIc to MSMP was first described
in~\cite{bradley-fmcad07,bradley-fac08}.)

\begin{proposition} \label{prop:fpic}
  $\tn{FPIc}\maps\tn{MSMP}$.
\end{proposition}

\begin{proof}
\proofitem{Reduction}
The reduction is defined as
follows. $\fml{R}\triangleq L(c)\triangleq\{l\,|\,l\in c\}$
and, 
\begin{equation} \label{eq:picpred}
  \Pred(\fml{W})\triangleq\neg\SAT(\fml{F}\land\land_{l\in\fml{W}}\,(\neg l))
\end{equation}
with $\fml{W}\subseteq\fml{R}$.
\proofitem{Monotonicity}
The predicate (see \eqref{eq:picpred}) is of form \pform, with
$\fml{G}\triangleq\fml{F}$, $u_i\triangleq l$, and
$\sigma(l)\triangleq\neg l$. Thus, by \autoref{prop:mono} the
predicate is monotone.
\proofitem{Correctness}
Let $\fml{M}$ be a minimal set such that $\Pred(\fml{M})$ holds.
Thus, $\fml{F}$ entails the literals in $\fml{M}$, and for any proper
superset $\fml{M}'$ of $\fml{M}$, $\fml{F}$ does not entail the
literals in $\fml{M}'$. Thus $\fml{M}$ is a prime implicate of
$\fml{F}$.
%
\end{proof}

Regarding FLEIt, $\fml{F}$ is in DNF, $\fml{F}=\lor_{j=1}^{m}\,t_j$,
and let the implicant to extend be $t_k$, with $1\le k\le m$.
The literals that can be used to extend $t_k$ are
$\fml{L}_t\triangleq\{l\in\LITS\,|\,\{l,\neg l\}\cap t_k=\emptyset\}$.
Moreover, let $\fml{D}=\lor_{j=1,j\not=k}^{m}\,t_j$.
Define $\fml{F}^{\tn{ItX}}\triangleq\fml{F}\land(\neg\fml{D})$, which
can be simplified to
$\fml{F}^{\tn{ItX}}\triangleq t_k\land\land_{i=1,i\not=k}^{m}\,(\neg t_i)$.

\begin{proposition} \label{prop:fleit}
  $\tn{FLEIt}\maps\tn{MSMP}$.
\end{proposition}

\begin{proof}(Sketch)
\proofitem{Reduction}
The reduction is defined as follows. $\fml{R}\triangleq\fml{L}_t$ and, 
\begin{equation} \label{eq:leitpred}
  \Pred(\fml{W})\triangleq\neg\SAT(\fml{F}^{\tn{ItX}}\land(\lor_{l\in\fml{R}\setminus\fml{W}}\,\neg l))
\end{equation}
with $\fml{W}\subseteq\fml{R}$.
\proofitem{Monotonicity}
The predicate (see \eqref{eq:leitpred}) is of form \bform, with
$\fml{G}\triangleq\fml{F}^{\tn{ItX}}$, $u_i\triangleq l$, and
$\sigma(l)\triangleq\neg l$. Thus, by \autoref{prop:mono} the
predicate is monotone.
\proofitem{Correctness} 
Let $u = t_k\land q$, where $q$ is a term.
Clearly, since $t_k\entails\fml{F}$, then $u\entails\fml{F}$,
and so $\fml{D}\lor u\entails\fml{F}$.
The issue is whether $\fml{F}\entails\fml{D}\lor u$, or equivalently
$\fml{F}\land(\neg\fml{D})\land(\neg u)\entails\bot$.
Expanding we get $t_k\land\land_{i=1,i\not=k}^{m}\,(\neg
t_i)\land(\neg t_k\lor \neg q)\entails\bot$, which can also be
simplified to $t_k\land\land_{i=1,i\not=k}^{m}\,(\neg
t_i)\land(\neg q)\entails\bot$.
Hence, the goal is to find a maximal set of literals $q$ such that
$t_k\land\land_{i=1,i\not=k}^{m}\,(\neg t_i)\land(\neg
q)\entails\bot$.
This can be converted to a minimization problem by removing all
literals and then adding to $q$ literals that can be included.
\end{proof}

Regarding FLEIc, $\fml{F}$ is in CNF, $\fml{F}=\land_{j=1}^{m}\,(c_j)$,
and let the implicate to extend be $c_k$, with $1\le k\le m$.
The literals that can be used to extend $c_k$ are
$\fml{L}_c\triangleq\{l\in\LITS\,|\,\{l,\neg l\}\cap c_k=\emptyset\}$.
Moreover, let $\fml{D}=\land_{j=1,j\not=k}^{m}\,(c_j)$.
Define $\fml{F}^{\tn{IcX}}\triangleq(\neg\fml{F})\land(\fml{D})$, which
can be simplified to
$\fml{F}^{\tn{IcX}}\triangleq (\neg c_k)\land\land_{i=1,i\not=k}^{m}\,(c_i)$.

\begin{proposition} \label{prop:fleic}
  $\tn{FLEIc}\maps\tn{MSMP}$.
\end{proposition}

\begin{proof}(Sketch)
\proofitem{Reduction}
The reduction is defined as follows. $\fml{R}\triangleq\fml{L}_c$ and, 
\begin{equation} \label{eq:leicpred}
  \Pred(\fml{W})\triangleq\neg\SAT(\fml{F}^{\tn{IcX}}\land(\lor_{l\in\fml{R}\setminus\fml{W}}\,l))
\end{equation}
with $\fml{W}\subseteq\fml{R}$.
\proofitem{Monotonicity}
The predicate (see \eqref{eq:leicpred}) is of form \bform, with
$\fml{G}\triangleq\fml{F}^{\tn{IcX}}$, $u_i\triangleq l$, and
$\sigma(l)\triangleq l$. Thus, by \autoref{prop:mono} the predicate is
monotone.
\proofitem{Correctness}  
Let $u = c_k\lor q$, where $q$ is a clause.
Clearly, since $\fml{F}\entails c_k$, then $\fml{F}\entails u$,
and so $\fml{F}\entails\fml{D}\land(u)$.
The issue is whether $\fml{D}\land(u)\entails\fml{F}$, or equivalently
$(\neg\fml{F})\land\fml{D}\land(u)\entails\bot$.
Expanding we get $(\neg
c_k)\land\land_{i=1,i\not=k}^{m}\,(c_i)\land(c_k\lor q)\entails\bot$,
which can also be simplified to $(\neg
c_k)\land\land_{i=1,i\not=k}^{m}\,(c_i)\land(q)\entails\bot$.
Hence, the goal is to find a maximal set of literals $q$ such that
$(\neg t_k\land\land_{i=1,i\not=k}^{m}\,(c_i)\land(q)\entails\bot$.
This can be converted to a minimization problem by removing all
literals and then adding to $q$ literals that can be included.
\end{proof}

\subsubsection{Formula Entailment}

\begin{proposition} \label{prop:fmnes}
  $\tn{FMnES}\maps\tn{MSMP}$.
\end{proposition}

\begin{proof}(Sketch)
\proofitem{Reduction}
The reduction is defined as follows. $\fml{R}\triangleq\fml{J}$ and,  
\begin{equation} \label{eq:mnespred}
  \Pred(\fml{W})\triangleq\neg\SAT(\neg\fml{I}\land\land_{c\in\fml{W}}\,(c))
\end{equation}
with $\fml{W}\subseteq\fml{R}$.
\proofitem{Monotonicity}
The predicate (see \eqref{eq:mnespred}) is of form \pform, with
$\fml{G}\triangleq\neg\fml{I}$, $u_i\triangleq c$, and
$\sigma(c)\triangleq c$. Thus, by \autoref{prop:mono} the predicate is
monotone.
\proofitem{Correctness}  
Simple, based on previous proofs.
\end{proof}

\begin{proposition} \label{prop:fmxes}
  $\tn{FMxES}\maps\tn{MSMP}$.
\end{proposition}

\begin{proof}(Sketch)
\proofitem{Reduction}
The reduction is defined as follows. $\fml{R}\triangleq\fml{N}$ and,  
\begin{equation} \label{eq:mxespred}
  \Pred(\fml{W})\triangleq\neg\SAT(\fml{J}\land(\lor_{c\in\fml{R}\setminus\fml{W}}\neg c))
\end{equation}
with $\fml{W}\subseteq\fml{R}$.
\proofitem{Monotonicity}
The predicate (see \eqref{eq:mxespred}) is of form \bform, with
$\fml{G}\triangleq\fml{J}$, $u_i\triangleq c$, and
$\sigma(c)\triangleq\neg c$. Thus, by \autoref{prop:mono} the
predicate is monotone.
\proofitem{Correctness}  
Simple, based on previous proofs.
\end{proof}

\subsubsection{Backbone Literals}

\begin{proposition} \label{prop:fbbr}
  $\tn{FBBr}\maps\tn{MSMP}$.
\end{proposition}

\begin{proof}
  \proofitem{Reduction}
  Consider a set of literals $\fml{V}$, obtained from an initial
  satisfying assignment $\nu$. The reduction is defined as follows.
  $\fml{R}\triangleq\fml{V}$ and,
  \begin{equation} \label{eq:bbrpred}
    \Pred(\fml{W})\triangleq\neg\SAT(\fml{F}\land(\lor_{l\in\fml{R}\setminus\fml{W}}\neg l))
  \end{equation}
  with $\fml{W}\subseteq\fml{R}$.
  \proofitem{Monotonicity}
  The predicate (see \eqref{eq:bbrpred}) is of form \bform, with
  $\fml{G}\triangleq\fml{F}$, $u_i\triangleq l$, and
  $\sigma(l)\triangleq\neg l$. Thus, by \autoref{prop:mono} the
  predicate is monotone.
  \proofitem{Correctness}
  Let $\fml{T}$ be a minimal set such that $\Pred(\fml{T})$ holds.
  By \autoref{def:mp}, since $\fml{T}$ is minimal for $\Pred$, then
  any proper superset $\fml{T}'$ of $\fml{T}$ the predicate does not
  hold, i.e.\ $\fml{F}\nentails\fml{V}\setminus\fml{T}'$.
  Hence, $\fml{F}\entails\fml{V}\setminus\fml{T}$, and for any proper
  superset $\fml{T}'$ of $\fml{T}$,
  $\fml{F}\nentails\fml{V}\setminus\fml{T}'$.
  Moreover, by \autoref{prop:unique}, $\fml{V}\setminus\fml{T}$ is
  maximal and unique.
  Thus, by \autoref{def:bbr}, $\fml{V}\setminus\fml{T}$ is the set of
  backbone literals of $\fml{F}$. 
  %
\end{proof}

\begin{observation}
  Regarding the reduction in the proof \autoref{prop:fbbr}, the
  computed minimal set $\fml{T}$ is the \emph{complement} of the set
  of backbone literals, which is given by $\fml{V}\setminus\fml{T}$.
\end{observation}

In practice, and for efficiency reasons, algorithms for computing the 
backbone of a Boolean formula start from a reference satisfying
assignment (i.e.\ the FBBr function problem)~\cite{jlms-aicomm13}.
However, the reduction of the general backbone computation problem to
MSMP yields interesting insights into the worst-case number of SAT
oracle queries needed to compute the backbone of a Boolean formula.

\begin{proposition} \label{prop:fbb}
  $\tn{FBB}\maps\tn{MSMP}$.
\end{proposition}

\begin{proof}(Sketch)
  \proofitem{Reduction}
  The reduction is defined as follows. $\fml{R}\triangleq X$, with
  $X=\var(\fml{F})$. Consider one auxiliary set of variables $X'$,
  with $|X'|=|X|$, and let:
  \begin{equation} \label{eq:bbpred}
    \fml{F}^{\tn{BB}}\triangleq\fml{F}[X/X]\land\fml{F}[X/X']
  \end{equation}
  Finally, let:
  \begin{equation}
    \Pred(\fml{W})\triangleq\neg\SAT(\fml{F}^{\tn{BB}}\land(\lor_{x\in\fml{R}\setminus\fml{W}}\,x\land\neg x'))
  \end{equation}
  \proofitem{Monotonicity}
  The predicate (see \eqref{eq:bbpred}) is of form \bform, with
  $\fml{G}\triangleq\fml{F}^{\tn{BB}}$, $u_i\triangleq x$, and
  $\sigma(x)\triangleq x\land x'$, where $x'$ is a new variable not in
  $\fml{R}$, but associated with $x$. Thus, by \autoref{prop:mono} the
  predicate is monotone.
  \proofitem{Correctness}  
  The predicate holds for $\fml{W}=\fml{R}$. Similarly to the FBBr
  case, $(\lor_{x\in\fml{R}\setminus\fml{W}}\,x\land\neg x')$ yields
  an empty clause.
  As for the FBBr case, the literals to be removed from $\fml{R}$ are
  the ones that are backbone literals since adding literals
  $(\lor_{l\in\fml{R}\setminus\fml{W}}\neg l)$ can only be done while
  keeping the formula unsatisfiable.
\end{proof}

\subsubsection{Variable Independence} 

\begin{proposition} \label{prop:fvindp}
  $\tn{FVInd}\maps\tn{MSMP}$ (form \pform).
\end{proposition}

\begin{proof}(Sketch)
  \proofitem{Reduction}
  Consider the original set of variables $X$ and another set of
  variables $Y$, such that $|Y| = |X|$.
  $\fml{F}$ is to be checked for equivalence against $\fml{F}[X/Y]$,
  i.e.\ a copy of itself using new variables. An additional constraint
  is that some of these variables are equivalent. The non-equivalence
  between the two formulas is captured as follows:
  \begin{equation}
    \fml{F}^{\tn{VInd}}\triangleq(\fml{F}[X/Y]\land\neg\fml{F}[X/X]\lor\neg\fml{F}[X/Y]\land\fml{F}[X/X])
  \end{equation}
  Given $\fml{F}^{\tn{VInd}}$, the reduction is defined as follows.
  $\fml{R}\triangleq X$ and,
  \begin{equation} \label{eq:vindpred}
    \Pred(\fml{W})\triangleq\neg\SAT(\fml{F}^{\tn{VInd}}\land\land_{x_i\in\fml{W}}\,(x_i\leftrightarrow y_i))
  \end{equation}
  \proofitem{Monotonicity}
  The predicate (see \eqref{eq:vindpred}) is of form \pform, with
  $\fml{G}\triangleq\fml{F}^{\tn{VInd}}$, $u_i\triangleq x_i$, and
  $\sigma(x_i)\triangleq x_i\leftrightarrow y_i$, where $y_i$ is a new
  variable not in $\fml{R}$, but associated with $x_i$. Thus, by
  \autoref{prop:mono} the predicate is monotone.
  \proofitem{Correctness}  
  For $\fml{W}=\fml{R}$ each $x_i$ variable is equivalent to
  its corresponding $y_i$ variable.
  A variable $x_i$ is removed from $\fml{W}$ if, by taking any
  possible value, it does not affect the equivalence of $\fml{F}$
  with $\fml{F}$ defined on $Y$ variables.
  In such a case, $\fml{F}$ is independent from $x_i$
\end{proof}

\subsubsection{Maximum Autarkies}

Autarkies can be captured with two different monotone predicate
forms. The first predicate is of form \lform.

\begin{proposition} \label{prop:fautl}
  $\tn{FAut}\maps\tn{MSMP}$ (form \lform).
\end{proposition}

\begin{proof}(Sketch)
  \proofitem{Reduction}
  The reduction uses a simplified version of the model proposed
  in~\cite{liffiton-sat08}.
  Consider a CNF formula $\fml{F}$ with a set of variables
  $X=\var(\fml{F})$. Create new sets of variables $X^{+}$, $X^{0}$,
  $X^{1}$, such that for $x\in X$, $x^{+}$ indicates whether a
  variable is selected, and $x^{1}$ and $x^{0}$ replace, respectively,
  the literals $x$ and $\neg x$ in the clauses of $\fml{F}$. Thus,
  $\fml{F}$ is transformed into a new formula $\fml{F}^{0,1}$, where
  each literal in $x$ is translated either into $x^{1}$ or $x^{0}$. 
  The resulting CNF formula, $\fml{F}^{\tn{Aut}}$, consists of
  CNF-encoding the following sets of constraints. For each $x\in X$,
  add to $\fml{F}^{\tn{Aut}}$ (the clauses resulting from encoding)
  $x^1 \leftrightarrow x^{+}\land x$ and $x^0 \leftrightarrow
  x^{+}\land\neg x$. If a clause $c^{0,1}\in\fml{F}^{0,1}$ has a
  literal in $x$, then add to $\fml{F}^{\tn{Aut}}$ the clause
  $(x^{+}\to c^{0,1})$.
  Now, $\fml{R}\triangleq X^{+}$ and,
  \begin{equation} \label{eq:autpredl}
    \Pred(\fml{W})\triangleq\SAT(\fml{F}^{\tn{Aut}}\land\land_{x^{+}\in\fml{R}\setminus\fml{W}}\,(x^{+}))
  \end{equation}
  with $\fml{W}\subseteq\fml{R}$.
  \proofitem{Monotonicity}
  The predicate (see \eqref{eq:autpredl}) is of form \lform, with
  $\fml{G}\triangleq\fml{F}^{\tn{Aut}}$, $u_i\triangleq x^{+}$, and
  $\sigma(x^{+}) = x^{+}$. Thus, by \autoref{prop:mono} the predicate
  is monotone.
  \proofitem{Correctness}  
  For $\fml{W}=\fml{R}$, predicate holds, since the argument
  of the SAT oracle call consists of $\fml{F}^{\tn{Aut}}$.
  The elements that are to be dropped from $\fml{R}$ are the
  variables for which there exists a truth assignment that
  identifies the maximum autarky, i.e.\ the autark variables.
  The minimal set is $\fml{W}$ = $\fml{R}\setminus\fml{K}$, where
  $\fml{K}$ is the set of autark variables.
  The monotonicity properties in this case can be used to
  prove that the set of autark variables is maximum.
\end{proof}

\begin{observation}
  Regarding the reduction in the proof \autoref{prop:fautl}, the
  computed minimal set $\fml{T}$ is the \emph{complement} of the set
  of autark variables, which is given by
  $\var(\fml{F})\setminus\fml{T}$.
\end{observation}

An alternative predicate for FAut is of form \bform, as shown next.

\begin{proposition} \label{prop:fautb}
  $\tn{FAut}\maps\tn{MSMP}$ (form \bform).
\end{proposition}

\begin{proof}(Sketch)
  \proofitem{Reduction}
  As before, the reduction uses a simplified version of the model
  proposed in~\cite{liffiton-sat08} (see proof of
  \autoref{prop:fautl}).
  Given these definitions, $\fml{R}\triangleq X^{+}$ and,
  \begin{equation} \label{eq:autpredb}
    \Pred(\fml{W})\triangleq\neg\SAT(\fml{F}^{\tn{Aut}}\land(\lor_{x^{+}\in\fml{R}\setminus\fml{W}}\,x^{+}))
  \end{equation}
  with $\fml{W}\subseteq\fml{R}$.
  \proofitem{Monotonicity}
  The predicate (see \eqref{eq:autpredb}) is of form \bform, with
  $\fml{G}\triangleq\fml{F}^{\tn{Aut}}$, $u_i\triangleq x^{+}$, and
  $\sigma(x^{+}) = x^{+}$. Thus, by \autoref{prop:mono} the predicate
  is monotone.
  \proofitem{Correctness}  
  Similar to previous proofs by showing that,
  \begin{equation}
    \fml{F}^{\tn{Aut}}\entails\land_{x^{+}\in\fml{R}\setminus\fml{W}}(\neg
    x^{+})
  \end{equation}
  %
\end{proof}

\subsection{Optimization Problems} \label{ssec:optim}

To illustrate the modeling flexibility of monotone predicates, this
section investigates how to solve optimization problems, namely FSMCS,
FSMDS, FLMFS and FSMnM. 
In contrast with the previous section, the objective here is not to
develop efficient algorithms or insight, but to show other uses of
monotone predicates. 
It should be noted that computing cardinality minimal sets for some of
the other function problems, e.g. FMUS, FMES, FPIt, FPIc, etc., is
significantly harder, since these function problems are in the second
level of the polynomial hierarchy,
e.g.~\cite{gupta-phd06,umans-tcad06}.

In the remainder of this section, unweighted formulations are
considered, i.e.\ each clause is soft with weight 1.

\begin{proposition} \label{prop:mapsmcs}
  $\tn{FSMCS}\maps\tn{MSMP}$.
\end{proposition}

\begin{proof}
\proofitem{Reduction}
Let $\fml{B} = \{0, 1, 2, \ldots, |\fml{F}|\}$ denote the possible
numbers of clauses that are required to be satisfied. Furthermore, for
each clause $c_i\in\fml{F}$, create a relaxed copy $(\neg p_i\lor
c_i)$ and let $\fml{F}^R$ denote the CNF formula where each clause
$c_i$ is replaced by its relaxed version. Let $\fml{P}$ denote the set
of selection variables. 
For each $b_j\in\fml{B}$ create the constraint
$\sum_{p_i\in\fml{P}}p_i\ge b_j$.
Each of these constraints defines a lower bound on the number of
satisfied clauses.
Define $\fml{R}\triangleq\fml{B}$. Moreover, given
$\fml{W}\subseteq\fml{R}$ let,
\begin{equation} \label{eq:qpred}
Q(\fml{W})\triangleq\bigwedge_{b_j\in\fml{R}\setminus\fml{W}}\left(\sum_{p_i\in\fml{P}}\,p_i\ge b_j\right)
\end{equation}
Finally, the predicate is defined as follows:
\begin{equation} \label{eq:mapsmcs}
  \Pred(\fml{W})\triangleq\SAT\left(\fml{F}^R\land Q(\fml{W})\right)
\end{equation}
\proofitem{Monotonicity}
Given the definition of $Q(\fml{W})$ in \eqref{eq:qpred}, the
predicate \eqref{eq:mapsmcs} is of form \lform, with
$\fml{G}\triangleq\fml{F}^R$, $u_i = b_j$, and
$\sigma(b_j)\triangleq\left(\sum_{p_i\in\fml{P}}\,p_i\ge b_j\right)$.
Thus, by \autoref{prop:mono} the predicate is monotone. 
\proofitem{Correctness}
The set of true $p_i$ variables picks a subset $\fml{S}$ of the
clauses in $\fml{F}$, which is to be checked for satisfiability.
Given the formulation, it holds that a subset-minimal set must
correspond to a cardinality minimal set.
If, by removing from $\fml{W}$ the element associated with some value
$b_j$ satisfies the predicate, then removing from $\fml{W}$ any
element associated with the $b_k$ of value no greater than $b_j$ also
satisfies the predicate, and this holds with the same truth
assignment, namely the same set of selected clauses.
Thus, any minimal set must exclude the element associated with the
largest value $b_j$ such that the predicate holds, and must also
exclude any element associated with $b_k$ of value no greater than
$b_j$.
Thus, the elements not removed from $\fml{W}$ are associated with
$b_k$ representing sizes of sets of clauses such that not all can be
simultaneously satisfied.
Hence, the minimal set represents the smallest MCS of $\fml{F}$.
\end{proof}

\begin{proposition} \label{prop:mapsmds}
  $\tn{FSMDS}\maps\tn{MSMP}$.
\end{proposition}

\begin{proof}(Sketch)
\proofitem{Reduction}
The definitions of the reduction in the proof of
\autoref{prop:mapsmcs} apply. As a result, the predicate is defined as
follows:
\begin{equation} \label{eq:mapsmds}
  \Pred(\fml{W})\triangleq\SAT\left(\neg\fml{F}\land\fml{F}^R\land
  Q(\fml{W})\right)
\end{equation}
\proofitem{Monotonicity}
Given the definition of $Q(\fml{W})$ in \eqref{eq:qpred}, the
predicate \eqref{eq:mapsmds} is of form \lform, with
$\fml{G}\triangleq\neg\fml{F}\land\fml{F}^R$, $u_i = b_j$, and
$\sigma(b_j)\triangleq\left(\sum_{p_i\in\fml{P}}\,p_i\ge b_j\right)$.
Thus, by \autoref{prop:mono} the predicate is monotone. 
\proofitem{Correctness} 
The set of true $p_i$ variables picks a subset $\fml{S}$ of the
clauses in $\fml{F}$. Thus, it must hold that
$\fml{F}\entails\fml{S}$.
Thus one wants to pick the largest subset $\fml{N}$ of $\fml{F}$ such
that $\fml{N}\nentails\fml{F}$.
A minimal set of $\fml{W}\subseteq\fml{R}$ corresponds to a maximal
set $\fml{R}\setminus\fml{W}$, containing all the constraints
$\left(\sum_{p_i\in\fml{P}}\,p_i\ge b_j\right)$ that can be
satisfied. Hence, this gives the largest set $\fml{N}$ of selected
clauses such that $\fml{N}\nentails\fml{F}$.
\end{proof}

\begin{proposition} \label{prop:mapsmcfs}
  $\tn{FSMCFS}\maps\tn{MSMP}$.
\end{proposition}

\begin{proof}(Sketch)
\proofitem{Reduction}
The elements to relax in this case are the complements of the clauses,
which will be satisfied (so that the clauses are falsified).
As a result, define $\fml{N}^R\triangleq \land_{c_i\in\fml{F}}
  (\neg p_i\lor\neg c_i)$. The predicate can then be defined as
follows: 
\begin{equation} \label{eq:mapsmcfs}
  \Pred(\fml{W})\triangleq\SAT\left(\fml{N}^R\land Q(\fml{W})\right)
\end{equation}
\proofitem{Monotonicity}
Given the definition of $Q(\fml{W})$ in \eqref{eq:qpred}, the
predicate \eqref{eq:mapsmcfs} is of form \lform, with
$\fml{G}\triangleq\fml{N}^R$, $u_i = b_j$, and
$\sigma(b_j)\triangleq\left(\sum_{p_i\in\fml{P}}\,p_i\ge b_j\right)$.
Thus, by \autoref{prop:mono} the predicate is monotone. 
\proofitem{Correctness}  
Similar to previous proofs in this section.
\end{proof}

\begin{proposition} \label{prop:mapsmnm}
  $\tn{FSMnM}\maps\tn{MSMP}$.
\end{proposition}

\begin{proof}(Sketch)
\proofitem{Reduction}
The elements to relax in this case are the variables assigned value
true. As a result, define $\fml{V}^R\triangleq \land_{x_i\in\fml{V}}
  (\neg p_i\lor\neg x_i)$. The predicate can then be defined as
follows: 
\begin{equation} \label{eq:mapsmnm}
  \Pred(\fml{W})\triangleq\SAT\left(\fml{F}\land\fml{V}^R\land Q(\fml{W})\right)
\end{equation}
\proofitem{Monotonicity}
Given the definition of $Q(\fml{W})$ in \eqref{eq:qpred}, the
predicate \eqref{eq:mapsmnm} is of form \lform, with
$\fml{G}\triangleq\fml{F}^R\land\fml{V}^R$, $u_i\triangleq b_j$, and
$\sigma(b_j)\triangleq\left(\sum_{p_i\in\fml{P}}\,p_i\ge b_j\right)$.
Thus, by \autoref{prop:mono} the predicate is monotone.
\proofitem{Correctness}  
Similar to previous proofs in this section.
\end{proof}

\subsection{Summary of Reductions} \label{ssec:genred}

\autoref{tab:maps} summarizes the reductions described in
\autoref{ssec:greds} and in \autoref{ssec:optim}. For each problem,
the associated reference provides the complete details of the
reduction to the MSMP problem.
\begin{table}[!t] 
  \begin{center}
    \renewcommand{\arraystretch}{1.275}
    \begin{tabular}{
        |>{\centering\arraybackslash}m{1.5cm}
        |>{\centering\arraybackslash}m{1.0cm}
        |>{\centering\arraybackslash}m{5.25cm}
        |>{\centering\arraybackslash}m{1.0cm}
        |>{\centering\arraybackslash}m{2.5cm}|}
      \hline
      Problem & $\fml{R}$ & $\Pred(\fml{W}),\fml{W}\subseteq\fml{R}$ & Form & Reference \\ \hline\hline
      FMUS & $\fml{F}$ & $\neg\SAT(\land_{c\in\fml{W}}\,(c))$ & \pform & Prop.~\ref{prop:fmus}, p.\ \pageref{prop:fmus} \\ \hline
      FMCS & $\fml{F}$ & $\SAT(\land_{c\in\fml{R}\setminus\fml{W}}\,(c))$ & \lform & Prop.~\ref{prop:fmcs}, p.\ \pageref{prop:fmcs} \\ \hline
      FMES & $\fml{F}$ & $\neg\SAT(\neg\fml{F}\land\land_{c\in\fml{W}}\,(c))$ & \pform & Prop.~\ref{prop:fmes}, p.\ \pageref{prop:fmes} \\ \hline
      FMDS & $\fml{F}$ & $\SAT(\neg\fml{F}\land\land_{c\in\fml{R}\setminus\fml{W}}\,(c))$ & \lform & Prop.~\ref{prop:fmds}, p.\ \pageref{prop:fmds} \\ \hline
      FCMFS & $\fml{F}$ & $\SAT(\land_{c\in\fml{R}\setminus\fml{W}}\,(\neg c))$ & \lform & Prop.~\ref{prop:fmcfs}, p.\ \pageref{prop:fmcfs} \\ \hline
      FMnM & $X$ & $\SAT(\fml{F}\land\land_{x\in\fml{R}\setminus\fml{W}}\,(\neg x))$ & \lform & Prop.~\ref{prop:fmnm}, p.\ \pageref{prop:fmnm} \\ \hline
      FPIt & $L(t)$ & $\neg\SAT(\neg\fml{F}\land\land_{l\in\fml{W}}\,(l))$ & \pform & Prop.~\ref{prop:fpit}, p.\ \pageref{prop:fpit} \\ \hline
      FPIc & $L(c)$ & $\neg\SAT(\fml{F}\land\land_{l\in\fml{W}}(\neg l))$ & \pform & Prop.~\ref{prop:fpic}, p.\ \pageref{prop:fpic} \\ \hline
      FLEIt & $\fml{L}_t$ & $\neg\SAT(\fml{F}^{\tn{ItX}}\land(\lor_{l\in\fml{R}\setminus\fml{W}}\,\neg l))$ & \bform & Prop.~\ref{prop:fleit}, p.\ \pageref{prop:fleit} \\ \hline
      FLEIc & $\fml{L}_c$ & $\neg\SAT(\fml{F}^{\tn{IcX}}\land(\lor_{l\in\fml{R}\setminus\fml{W}}\,l))$ & \bform & Prop.~\ref{prop:fleic}, p.\ \pageref{prop:fleic} \\ \hline
      FMnES & $\fml{J}$ & $\neg\SAT(\neg\fml{I}\land\land_{c\in\fml{W}}\,(c))$ & \pform & Prop.~\ref{prop:fmnes}, p.\ \pageref{prop:fmnes} \\ \hline
      FMxES & $\fml{N}$ & $\neg\SAT(\fml{J}\land(\lor_{c\in\fml{R}\setminus\fml{W}}\neg c))$ & \bform & Prop.~\ref{prop:fmxes}, p.\ \pageref{prop:fmxes} \\ \hline
      FBBr & $\fml{V}$ & $\neg\SAT(\fml{F}\land(\lor_{l\in\fml{R}\setminus\fml{W}}\neg l))$ & \bform & Prop.~\ref{prop:fbbr}, p.\ \pageref{prop:fbbr} \\ \hline
      FBB & $X$ & $\neg\SAT(\fml{F}^{\tn{BB}}\land(\lor_{x\in\fml{R}\setminus\fml{W}}\,x\land\neg x'))$ & \bform & Prop.~\ref{prop:fbb}, p.\ \pageref{prop:fbb} \\ \hline
      FVInd & $X$ & $\neg\SAT(\fml{F}^{\tn{VInd}}\land\land_{x_i\in\fml{W}}\,(x_i\leftrightarrow y_i))$ & \pform & Prop.~\ref{prop:fvindp}, p.\ \pageref{prop:fvindp} \\ \hline
      %
      FAut & $ X^{+}$ & $\SAT(\fml{F}^{\tn{Aut}}\land\land_{x^{+}\in\fml{R}\setminus\fml{W}}\,(x^{+}))$ & \lform & Prop.~\ref{prop:fautl}, p.\ \pageref{prop:fautl} \\ \hline
      FAut & $ X^{+}$ & $\neg\SAT(\fml{F}^{\tn{Aut}}\land(\lor_{x^{+}\in\fml{R}\setminus\fml{W}}\,x^{+}))$ & \bform & Prop.~\ref{prop:fautb}, p.\ \pageref{prop:fautb} \\
      \hline\hline
      %
      %
      %
      FSMCS  & $\fml{F}$ & $\SAT\left(\fml{F}^R\land Q(\fml{W})\right)$ & \lform & Prop.~\ref{prop:mapsmcs}, p.\ \pageref{prop:mapsmcs} \\ \hline
      FSMDS  & $\fml{F}$ & $\SAT\left(\neg\fml{F}\land\fml{F}^R\land Q(\fml{W})\right)$  & \lform & Prop.~\ref{prop:mapsmds}, p.\ \pageref{prop:mapsmds} \\ \hline
      FSMCFS & $\fml{F}$ & $\SAT\left(\fml{N}^R\land Q(\fml{W})\right)$ & \lform & Prop.~\ref{prop:mapsmcfs}, p.\ \pageref{prop:mapsmcfs} \\ \hline
      FSMnM  & $X$ & $\SAT\left(\fml{F}\land\fml{V}^R\land Q(\fml{W})\right)$ & \lform & Prop.~\ref{prop:mapsmnm}, p.\ \pageref{prop:mapsmnm} \\ \hline
    \end{tabular}
    \caption{Overview of function problems and MSMP reductions.
      Additional problems include FMSS, FMNS, FMFS, FMxM,
      FLMSS, FLMNS, FLMFS and FLMxM. Other function problems, with
      reductions to MSMP similar to the ones above, are covered in 
      \autoref{sec:probs}.
    } \label{tab:maps}
  \end{center}
\end{table}
The first part denotes function problems where the goal is to compute
a minimal set. The second part denotes function problems that
represent optimization problems.
Besides the function problems summarized in~\autoref{tab:maps},
\autoref{ssec:fdefs} and also \autoref{ssec:greds} indicate that
several other problems related with Boolean formulas can also be
mapped into the MSMP problem. These include variants of the presented
problems when considering groups of clauses, variables, hard clauses,
etc.

It is important to observe that, for some problems, there are
dedicated algorithms that require an asymptotically smaller number of
queries to a SAT oracle than the most efficient of the MSMP
algorithms described in the next section. Nevertheless, as shown above
the reduction to MSMP yeilds new alternative algorithms and, for a
number of cases, it allows developing relevant new insights.


%
%
%

\section{Conclusions \& Research Directions} \label{sec:conc}

This paper extends recent
work~\cite{bradley-fmcad07,bradley-fac08,msjb-cav13} on monotone
predicates and shows that a large number of function problems defined
on Boolean formulas can be reduced to computing a minimal set over a
monotone predicate. 
The paper also argues that monotone predicates find application in
more expressive domains, including ILP, SMT and CSP.

A number of research directions can be envisioned.
A natural question is to identify other function problems that can be
reduced to the MSMP problem.
Algorithms for MSMP are described elsewhere~\cite{msjb-cav13}. A
natural research question is whether additional algorithms can be
developed.
In addition, most practical algorithms for solving minimal set
problems exploit a number of pruning
techniques~\cite{blms-aicomm12,msjb-cav13,mshjpb-ijcai13}. Another
natural research question is whether these techniques can be used in
the more general setting of MSMP.
Finally, another line of research is to develop precise query
complexity characterizations of instantiations of the MSMP.
Concrete examples include FMUS, FPIc, FPIt, FBB, FVInd, among others.


%
%
%

\subsubsection*{Acknowledgments.}
\noindent
This work was influenced by many discussions with colleagues at
CASL/UCD and at IST/INESC-ID. 
This work is partially supported by SFI PI grant BEACON
(09/IN.1/I2618), and by FCT grants ATTEST (CMU-PT/ELE/0009/2009),
POLARIS (PTDC/EIA-CCO/123051/2010), and INESC-ID's multiannual PIDDAC 
funding PEst-OE/\-EEI/\-LA0021/\-2011.


\bibliographystyle{abbrv}
\bibliography{msmp}
\end{document}